\documentclass[11pt]{article}

\usepackage{amssymb,amsthm,amsmath}
\usepackage{nicefrac}

\usepackage[usenames,dvipsnames]{xcolor}
\usepackage[colorlinks,citecolor=blue,linkcolor=BrickRed]{hyperref}\usepackage[colorlinks,citecolor=blue,linkcolor=BrickRed]{hyperref}
\usepackage{thmtools,thm-restate}
\usepackage{comment}
\usepackage{cleveref}
\usepackage{fullpage}
\usepackage{graphicx}

\newtheorem{theorem}{Theorem}[section]
\newtheorem{corollary}[theorem]{Corollary}

\newtheorem{lemma}[theorem]{Lemma}
\newtheorem{observation}[theorem]{Observation}
\newtheorem{definition}[theorem]{Definition}

\newtheorem{remark}[theorem]{Remark}
\newtheorem{claim}[theorem]{Claim}

\newtheorem*{claim*}{Claim}

\newcommand{\poly}{\ensuremath{\mathsf{poly}}}
\newcommand{\E}{\mathbb{E}}
\newcommand{\Var}{\mathsf{Var}}
\newcommand{\R}{\mathbb{R}}
\newcommand{\C}{\mathbb{C}}
\newcommand{\Z}{\mathbb{Z}}

\newcommand{\ind}{\mathbf{1}}

\newcommand{\vol}{\mathsf{vol}}

\newcommand{\disc}{\mathsf{disc}}

\newcommand{\RHS}{\mathsf{RHS}}

\newcommand{\ignore}[1]{{}}
\newcommand{\eps}{\epsilon}

\newcommand{\err}{\mathsf{err}}

\newcommand{\SO}{\mathsf{SO}}
\newcommand{\HK}{\mathsf{HK}}

\newcommand{\D}{\mathcal{D}}

\newcommand{\nnz}{\mathsf{nnz}}

\newcommand{\SubgVecBal}{\mathsf{BalSubgDisc}}
\newcommand{\SubgSparsification}{\mathsf{SubgTransference}}

\newcommand{\bfk}{\mathbf{k}}

\newcommand{\bfz}{\mathbf{z}}
\newcommand{\bfv}{\mathbf{v}}
\newcommand{\bfx}{\mathbf{x}}
\newcommand{\bfu}{\mathbf{u}}
\newcommand{\bfd}{\mathbf{d}}

\newcommand{\bfs}{\mathbf{s}}
\newcommand{\bfj}{\mathbf{j}}

\allowdisplaybreaks

{\hspace*{\fill}$\Box$\par}

\title{Quasi-Monte Carlo Beyond Hardy-Krause}

\author{Nikhil Bansal\thanks{University of Michigan, Ann Arbor, MI, USA. \texttt{bansal@gmail.com}. Supported in part by  NWO VICI grant 639.023.812 and  NSF award CCF-2327011.} 
 \and 
Haotian Jiang\thanks{University of Chicago, Chicago, IL, USA. \texttt{jhtdavid96@gmail.com}.}}

\date{}

\begin{document}

\begin{titlepage}
  
\maketitle
 \begin{abstract}
 \thispagestyle{empty}
 We examine the problem of numerically estimating the integral of a function $f$. The classical approaches to this problem are Monte Carlo (MC)  and quasi-Monte Carlo (QMC) methods. MC methods use random samples to evaluate $f$ and have error $O(\sigma(f)/\sqrt{n})$, where $\sigma(f)$ is the standard deviation of $f$. QMC methods are based on evaluating $f$ at explicit point sets with low discrepancy, and as given by the classical Koksma-Hlawka inequality, they have error $\widetilde{O}(\sigma_{\HK}(f)/n)$, where $\sigma_{\HK}(f)$ is the variation of $f$ in the sense of Hardy and Krause. These two methods have distinctive advantages and shortcomings, and a fundamental question is to find a method that combines the advantages of both. 

\smallskip

In this work, we give a simple randomized algorithm that produces QMC point sets with the following desirable features:
\begin{enumerate}
    \item It achieves substantially better error than given by the classical Koksma-Hlawka inequality. In particular, it has error $\widetilde{O}(\sigma_{\SO}(f)/n)$, where $\sigma_{\SO}(f)$ is a new measure of variation that we introduce, which is substantially smaller than the Hardy-Krause variation.     
    \item The algorithm only requires random samples from the underlying distribution, which makes it as flexible as MC.
    \item It automatically achieves the best of both MC and QMC (and the above improvement over Hardy-Krause variation and Koksma-Hlawka inequality) in an optimal way.  
    \item The algorithm is extremely efficient, with an amortized $\widetilde{O}(1)$ runtime per sample. 
\end{enumerate}
Our method is based on the classical transference principle in geometric discrepancy, combined with recent algorithmic innovations in combinatorial discrepancy \cite{bdgl18,ALS21,HSSZ24}, that besides producing low-discrepancy colorings, also guarantee certain subgaussian properties. This allows us to bypass several limitations of previous works in bridging the gap between MC and QMC methods and go beyond the Hardy-Krause variation. 

\smallskip

In particular, we show how to leverage subgaussianity and other structural properties of the discrepancy problem arising from the transference principle to obtain certain careful cancellations in the exact expression for the numerical integration error, which leads to our improved notion of variation. In contrast, these cancellations are completely lost in the Koksma-Hlawka inequalities and have not been exploited in previous works.

\end{abstract}

\end{titlepage}
{\hypersetup{linkcolor=BrickRed}
\tableofcontents
}
 \thispagestyle{empty}

\maketitle

 \clearpage
\setcounter{page}{1}
 \allowdisplaybreaks
\section{Introduction}
\label{sec:intro}
The problem of numerical integration is ubiquitous in science and engineering. Without loss of generality, this is equivalent to integrating a function $f$ over the unit cube,\footnote{For arbitrary domains and densities, one can apply standard transformations and subsume them in the integrand $f$.
See \cite{DP10,Owe13} and the references therein for more details.} i.e., computing $\overline{f} := \int_{[0,1]^d} f(\bfx) d \bfx$. 
Often, an exact computation is impossible as $f$ might be too complicated or may only be accessible via evaluations. 
A classical approach for estimating $\overline{f}$ is the Monte Carlo method (MC) -- pick $n$ i.i.d. uniformly random samples $\bfx^1,\cdots, \bfx^n \in [0,1]^d$ and output the estimate $\overline{f}(X) := \sum_j f(\bfx^j)/n$. 
As is well-known, the estimation error $\eps = |\overline{f}(X)-\overline{f}|$ scales as
\begin{equation}
    \label{eq:mc-error}
 \eps\approx \sigma(f)/\sqrt{n},
\end{equation} where $\sigma(f)$ is the standard deviation of $f$, i.e., we need $n \approx \sigma(f)^2/\eps^2$ for accuracy $\eps$.
As Monte Carlo methods only require random samples, they are powerful, flexible and widely applicable, and have been considered to be among the top 10 algorithms of the 20th century \cite{Dongarra2000}.



In practice, obtaining samples or even evaluating $f$ is often expensive -- e.g., a sample may be a random person from a population and $f$ the output of a medical treatment.
For this reason, techniques to improve the  tradeoff in \eqref{eq:mc-error} between $n$ and $\eps$ 
have been studied extensively.
Broadly, one approach is to reduce $\sigma(f)$ using variance reduction techniques such as importance sampling, stratified sampling and antithetic variables \cite{HH64, RC04, Liu01}. These methods usually require some information about $f$.
See the textbooks \cite{Gla04,Fis06} for more details. 

\smallskip
\noindent \textbf{Quasi-Monte Carlo Methods.}
Another approach, and our focus here, 
are quasi-Monte Carlo (QMC) methods, that target the convergence rate with $n$ and improve it to $\widetilde{O}_d(1/n)$.\footnote{Throughout the paper, $\widetilde{O}(\cdot)$ hides $O(\poly(\log n))$ factors; and $\widetilde{O}_d(\cdot)$ hides $O(\poly(\log^{O(d)} n))$ factors.} To achieve this, they abandon the flexibility of using random samples and instead  output $\overline{f}(X) = \sum_j f(\bfx^j)/n$ for a set of carefully-chosen \emph{deterministic} points $X=\{\bfx^1,\cdots,\bfx^n\}$. 

These methods can give large speedups for moderate dimensions, typically $d \leq 10$ \cite{Sloan1998, Dick2013}, and are used extensively in finance \cite{Gla04, Cools2015}, physics \cite{Binder2010, Dick2013}, graphics \cite{Keller2006, Dutre2006} and machine learning \cite{Gerber2015, Bergstra2012}. 
We describe these briefly next, but refer to the texts \cite{Lem09, Dick2013,Owe13} for details, and to \cite{Mat09} for a more theory-friendly introduction.

The core idea underlying QMC methods is an exact characterization of the numerical integration error for any set $X$ and any  function $f$ \cite{Kok42,Hla61,Zar68}. For illustrative purposes, let us consider 1-d (see \Cref{lem:Hlawka_Zaremba_formula} for general $d$). Here, the Hlawka-Zaremba formula\footnote{The formula \eqref{eq:1d-qmc-error} requires mild smoothness assumptions.} states that the error is exactly
\begin{equation}
\label{eq:1d-qmc-error}
     \eps = \frac{1}{n} \sum_{j=1}^n f(x^j) - \int_0^1 f(x) dx = \frac{1}{n} \int_0^1  D(x) f'(x) dx,
\end{equation} 
where $D(x) := nx - |[0,x] \cap X]|$ measures the ``continuous discrepancy" between the size of the interval $[0,x]$ and the number of points $x^1,\cdots,x^n$ contained in it. 

\smallskip
\noindent {\bf Koksma-Hlawka Inequalities.}
This directly gives several classical Koksma-Hlawka inequalities, 
that bound the error by the continuous discrepancy of $X$ and certain measures of variations of $f$.
E.g., in $1$-d, applying Cauchy-Schwartz to \eqref{eq:1d-qmc-error} gives the following $\ell_2$-Koksma-Hlawka inequality:\footnote{The $\ell_2$-Koksma-Hlawka inequality is also known as Zaremba's inequality.} 
\begin{equation} \label{eq:1d_ell2_koksma-hlawka}
\eps \leq \frac{1}{n} \Big(\int_0^1 D(x)^2 d x \Big)^{1/2} \cdot \Big(\int_0^1 |f'(x)|^2 d x \Big)^{1/2} = \frac{1}{n} \cdot D_2^*(X) \cdot \sigma_{\HK}(f) .
\end{equation}
Here, $D_2^*(X) := (\int_0^1 D(x)^2 d x)^{1/2}$ is the $\ell_2$-discrepancy for prefixes, and
$\sigma_{\HK}(f) := (\int_0^1 |f'(x)|^2 d x)^{1/2}$ is the $\ell_2$ variation of $f$ in the sense of Hardy and Krause (see \Cref{subsec:geo_disc_num_int} for general $d$). Similarly, one  has the $\ell_\infty/\ell_1$ version of the Koksma-Hlawka inequality
\begin{equation}
\label{eq:vhk}
\eps \leq (1/n) \cdot D^*(X) \cdot V_{\HK}(f),
\end{equation}
 where $D^*(X) := \max_x |D(x)|$ is the $\ell_\infty$-continuous discrepancy (or {\em star} discrepancy) and $V_{\HK}(f) := \int_0^1 |f'(x)| d x$ is the $\ell_1$-Hardy-Krause variation.\footnote{We consider both these versions later, but the $\ell_2$-version is cleaner to work with as the $\ell_2$-Hardy-Krause variation admits an explicit formula in terms of the Fourier coefficients of $f$, which allows for clean comparisons.}
 
 These Koksma-Hlawka inequalities are sharp \cite{WikiLowDisc,Mat09, KN12,Owe13} in the sense that for any point set $\{x_1,\ldots,x_n\}$ there is a function $f$ for which they hold with equality. This suggests that the only room for getting better bounds might be by improving the discrepancy terms.

Indeed, most of the works on QMC methods have been on designing low-discrepancy point sets $X$, 
and starting from the seminal work of van der Corput \cite{vdC35a,vdC35b}, numerous ingenious constructions have been found.
These constructions are often very delicate and subtle, 
with the best-known bounds being $\widetilde{O}_d(1)$ for all versions of continuous discrepancy. 
We discuss these in some more detail in \Cref{subsec:related_work}, but refer the reader to the  excellent textbooks \cite{Nie92, DP10, Mat09, Owe13} for these constructions and the fascinating mathematics and history behind them.

\smallskip
\noindent \textbf{Limitations of QMC.} Despite its tremendous success, however, ``QMC is not a panacea ..." \cite{WikiQMC} and it also has drawbacks and limitations. The most outstanding ones include: 
\begin{enumerate}
    \item \label{limitation:large_HK} The variation $\sigma_{\HK}(f)$ in \eqref{eq:1d_ell2_koksma-hlawka} can be much larger than $\sigma(f)$ in \eqref{eq:mc-error}, and hence QMC may perform much worse than MC in certain regimes.
    E.g., for $f(x)=\sin(kx)$, it is easily checked that $\sigma(f) = \Theta(1)$ while $\sigma_{\HK}(f) = \Theta(k)$ (as $f'= k \cos kx$). 
    Hence QMC is outperformed by MC when $k = \Omega(\sqrt{n})$. 
    \item \label{limitation:random_points} In many applications, one might only have access to random samples, and it may be  impossible to pick the specific points required by explicit QMC constructions.\footnote{Recall the example where samples are people. Here, choosing the person corresponding to a specific point $x$ is infeasible, and typically one can only sample  random subjects from a population.} Consequently, this significantly limits the flexibility and applicability of QMC. 

    \item Using the same deterministic point sets for different experiments might suffer from worst-case outcomes;  more problematically, they don't allow for error estimates statistically.
\end{enumerate}

There have been numerous attempts to address these limitations 
by considering 
 {\em Randomized QMC methods}, detailed in \cite{Lem09, Mat09, DP10, Owe13}. However, in contrast to the remarkable progress in constructing low-discrepancy QMC sets, the success here has been relatively limited. We discuss these works below in the context of our results and also in \Cref{subsec:related_work}.

\subsection{Our Contribution and Results}
We give a randomized QMC construction that, surprisingly, goes beyond the Hardy-Krause variation and achieves substantially lower integration error than given by the Koksma-Hlawka inequality. 

Our method also has various other desirable properties: (i) it is fast, (ii) it uses random samples and thus preserves the flexibility of MC, and (iii) it is never worse than MC, and in fact can have substantially lower error than both MC and QMC bounds.
This overcomes the limitations stated above, and even more. 
We now describe the algorithm and the  results.

\smallskip \noindent \textbf{The Algorithm.}
Our algorithm $\SubgSparsification$ starts with $n^2$ uniformly random samples 
$A_0 \subseteq [0,1]^d$, and partitions $A_0$ into $n$ sets $A_T^{(1)},\cdots,A_T^{(n)}$, each of size $n$. This is done by recursively applying an algorithm for combinatorial discrepancy (discussed in \Cref{subsec:subg_disc}).
A formal presentation of the algorithm is given in \Cref{sec:subg_sparsification_alg}.  

The overall algorithm runs in time $\widetilde{O}_d(n^2)$, and thus it requires only $\widetilde{O}_d(1)$ amortized time to produce each output point.
Each of these output sets $A_T^{(i)}$, which we denote by $A_T$  for short, 
satisfies various guarantees that we describe in \Cref{thm:main_high_dim,thm:best-of-both-1,thm:best-of-both-2} below.

\subsubsection{Beyond Hardy-Krause Variation}
\label{subsubsec:beyond_hardy-krause}
We introduce a new notion of variation of a function $f$ that we call the ``smoothed-out variation", denoted by $\sigma_{\mathsf{SO}}(f)$.  We describe this later below and also show how it is substantially better than the  Hardy-Krause variation $\sigma_{\HK}(f)$. For example, it satisfies $\sigma_{\SO}(f) \leq \sqrt{\sigma(f) \sigma_{\HK}
(f)}$.

Our first main result is that the
QMC point sets $A_T$ produced by our algorithm substantially improve upon the error given by the Koksma-Hlawka inequalities, and satisfy the following guarantee.

\begin{restatable}
{theorem}{IntegralSubgHighDim} \label{thm:main_high_dim}
For {\em every} function $f: [0,1]^d \rightarrow \R$, the integration error using points in $A_T$ satisfies
\[
\E[(\err(A_T,f))^2] \leq  \widetilde{O}_d(1) \cdot \Big(\frac{\sigma_{\mathsf{SO}}(f)^2}{n^2} \Big).
\]
That is, for any $f$ the typical error is $\widetilde{O}_d(\sigma_{\mathsf{SO}}(f)/n)$. 
\end{restatable}

This result is rather surprising, as Koksma-Hlawka inequalities are tight in general. Even allowing randomization in the point set, there seems no obvious room for improvement. This was precisely the reason that most previous works on QMC methods, 
including those on randomized ones (discussed later in \Cref{subsubsec:best-of-both-results}),
inherently start off with the Koksma-Hlawka inequalities and Hardy-Krause variations, and focus on constructing low-discrepancy QMC sets.\footnote{We are only aware of the works by Owen \cite{Owe97b,Owe08} which are not based on Koksma-Hlawka inequalities. However, the focus there is different, and the notion of variation used there can be arbitrarily worse than $\sigma_{\HK}$.}

Our key idea to go beyond this inherent bottleneck is to work directly with the Hlawka-Zaremba formula \eqref{eq:1d-qmc-error}, and exploit  some subtle cancellations inside the integral. Observe that in contrast, the Koksma-Hlawka inequality \eqref{eq:Koksma_Hlawka_ineq} loses any possibility of exploiting such cancellations.

To do this, we carefully combine subgaussianity properties of the low combinatorial discrepancy colorings with the combinatorial structures of the transformations arising in the transference principle, and use additional analytic observations. We give a high level overview in Section \ref{sec:overview}.

\smallskip 
\noindent \textbf{The Variation $\sigma_{\mathsf{SO}}(f)$.} 
We now define the variation $\sigma_{\mathsf{SO}}(f)$.
A convenient way is to define it  using Fourier analysis, which allows for clean comparisons with $\sigma(f)$ and $\sigma_{\HK}(f)$.\footnote{Later we also define it differently without Fourier analysis, see \Cref{remark:general_form}. This allows for comparison with $V_{\HK}$.}
Let \[f(\bfz) = \sum_{\bfk \in \Z^d} \widehat{f}(\bfk) \cdot \exp{(2 \pi i \langle \bfk, \bfz\rangle)}.\]
be the Fourier series expansion of $f$.
By a standard calculation (see Section \ref{subsec:Fourier_analysis}), $\sigma(f)$ and $\sigma_{\HK}(f)$  can be written as
\begin{align} 
\sigma(f)^2 & =  \sum_{\bfk \in \mathbb{Z}^d: \bfk\neq \mathbf{0}} \big|\widehat{f}(\bfk) \big|^2
\quad \text{and}, \label{eq:MC_formula_intro}\\
 \sigma_{\HK}(f)^2 & =  \Theta_d(1) \sum_{\bfk \in \Z^d: \bfk\neq \mathbf{0}} 
\big|\widehat{f}(\bfk) \big|^2 \cdot \prod_{j \in [d]} \max(1,|\bfk_j|^2).
\label{eq:MC_HK_formula_intro}
\end{align}
We define the  smoothed-out variation $\sigma_{\SO}(f)$ as
\begin{align} \label{eq:sigma_SO_formula_intro}
\sigma_{\SO}(f)^2 := 
\sum_{\bfk\in \Z^d: \bfk\neq \mathbf{0}} 
 \big|\widehat{f}(\bfk) \big|^2 \cdot \prod_{j \in [d]} \max(1,|\bfk_j|).
\end{align}
Notice that the $|\bfk_j|^2$ terms in the expression for $\sigma_{\HK}$ in \eqref{eq:MC_HK_formula_intro} are replaced by $|\bfk_j|$.

For better intuition for these, consider the $1$-d  case of $f(x) = \sum_{k \in \Z} \widehat{f}(k) \cdot \exp{(2 \pi i k x)}$. Here, 
\[\sigma(f)^2 =  \sum_{k \in \Z \setminus \{0\}} \big| \widehat{f}(k) \big|^2, \text{ and } \quad  \sigma_{\HK}(f)^2 = 
\int_0^1 |f'(x)|^2 d x = 
\sum_{k \in \Z \setminus \{0\}} 4\pi^2 k^2 \cdot \big| \widehat{f}(k) \big|^2,\]
while
\[\sigma_{\SO}(f)^2  = 
\sum_{k \in \Z \setminus \{0\}} k \cdot \big| \widehat{f}(k) \big|^2 .\]
Clearly, these formulas imply that $\sigma(f) \leq \sigma_{\SO}(f) \leq \sigma_{\HK}(f)$, with larger gaps for functions with higher frequencies (i.e., larger $k$). Further, by the Cauchy-Schwartz inequality, we have  \[\sigma_{\SO}(f)^2 \leq \sigma(f) \cdot \sigma_{\HK}(f).\] 
Thus, in a sense, the new measure $\sigma_{\SO}(f)$ bridges the standard deviation in MC and the Hardy-Krause variation in QMC by smoothing out the high-frequency terms (hence our name for it).

\begin{remark}[Smoothness Assumptions]
\label{remark:smoothness}
Requiring that $f$ has a Fourier series uses some (mild) smoothness assumptions.\footnote{It also requires $1$-periodicity, but this is without loss of generality via standard folding.} These can be removed by exploiting that our method uses random samples and that the Fourier series of $f$ converges to $f$ in a $\ell_2$ sense. We ignore these technicalities as they are standard and do not give any new insight.
\end{remark}

\begin{remark}[General Form of $\sigma_{\SO}$]
\label{remark:general_form} Our proof of Theorem \ref{thm:main_high_dim} yields another (more general) form of $\sigma_{\SO}$ that does not involve Fourier series. Using this, we can give an analogous improvement to the $\ell_\infty/\ell_1$-Koksma-Hlawka inequality. We defer this discussion to \Cref{subsec:general_def_sigma_SO}, as this is harder to describe and compare due to lack of an explicit formula for the $\ell_1$-Hardy-Krause variation $V_{\HK}(f)$.
\end{remark}

Concretely, \Cref{thm:main_high_dim} already gives substantial improvements over the usual QMC bounds. 

\smallskip
\noindent{\bf Example.}
Consider the example of $f(x) = \sin (kx)$ from before, where the MC error is $1/\sqrt{n}$ and the standard QMC error is $\approx k/n$. In contrast, the error in Theorem \ref{thm:main_high_dim} scales as $\sqrt{k}/n$, which is significantly better than the QMC bound. In particular, standard QMC beats MC only for $k\leq \sqrt{n}$, while our approach beats MC for $k \leq n$.

\subsubsection{Achieving Best of Both Worlds}
\label{subsubsec:best-of-both-results}
Our construction has another useful property, which we call best of both worlds, in the sense that it combines the benefits of MC and QMC methods in the optimal way.

We first state a weaker variant that combines MC and standard QMC bounds. 
Consider the sets $A_T$ produced by our algorithm. We have the following guarantee.
\begin{restatable}[Best of MC and Hardy-Krause]{theorem}{BestofBothWorlds} \label{thm:best-of-both-1}
For any function $f: [0,1]^d \rightarrow \R$, its integration error using $A_T$ is unbiased, i.e.,~$\E[\err(A_T,f)] = 0$, and satisfies that 
\[
\E[(\err(A_T,f))^2] \leq \widetilde{O}_d(1) \cdot \min_{f = g + h} \Big( \frac{\sigma(g)^2}{n} + \frac{V_{\HK}(h)^2}{n^2} \Big) ,
\]
where the infimum is over all decompositions of $f$ as the sum of two functions $g,h: [0,1]^d \rightarrow \R$. 
\end{restatable}
Consider the following instructive example which shows that this error can be much better than both MC and QMC.

\smallskip
\noindent {\bf Example.}  Consider  $f(x) = \sin (x) + k^{-1/2} \sin (kx)$. Then $\sigma(f) \approx 1$, and as $f'(x) = \cos(x) + \sqrt{k} \cos(kx)$, we have $V_{\HK}(f) \approx \sigma_{\HK}(f) \approx \sqrt{k}$. 
So the MC and QMC errors scale as $1/\sqrt{n}$ and $\sqrt{k}/n$ respectively. In particular for $k=n$, both are about $1/\sqrt{n}$. 

Now let us decompose $f=g+h$ with $g= k^{-1/2} \sin (kx)$ and $h=\sin (x)$. Then we have $\sigma(g) \approx O(1/\sqrt{k})$ and $V_{\HK}(h) \approx \sigma_{HK}(h) \approx O(1)$. Thus, by  \Cref{thm:best-of-both-1}, the  error $\approx (1/\sqrt{kn} + 1/n)$. In  particular for $k = n$, \Cref{thm:best-of-both-1} has substantially lower error $\approx 1/n$, even though both MC and QMC had error $\approx 1/\sqrt{n}$, as we saw above. 

\smallskip 
\noindent{\bf Decomposition Obliviousness.}
In practice, it is often recommended to use QMC together with MC methods \cite{Lem09}. In fact, one suggested
heuristic is exactly to first decompose $f= g+h$ suitably by hand, and use MC for $g$ and QMC for $h$. 
But this requires understanding the structure of $f$, which might not be possible in many scenarios. In practice, combining QMC and MC to get large speedups is often an art.

In contrast, notice that our method in Theorem \ref{thm:best-of-both-1} requires {\em no} explicit knowledge of the decomposition $f=g+h$ at all. That is, simply evaluating $f$ on the points $A_T$ automatically achieves the guarantee using the best possible decomposition of $f$.

\smallskip 
\noindent{\bf Previous Works.} There has been huge interest in designing Randomized QMC methods that combine the benefits of both MC and QMC, e.g. they enable statistical estimation of error by repeated experiments and protect against  worst-case outcomes. See \cref{subsec:related_work} 
and the excellent textbooks \cite{Lem09,Mat09,Owe13} for details. 
In our context, the bound in \Cref{thm:best-of-both-1} can also be achieved using variants of Owen's scrambled nets \cite{Owe97a,Mat98c}.   
However, these constructions are rather delicate and involved -- based on carefully applying permutations to digits of explicit initial point sets, known as digital nets \cite{Sob67,Nie92}, and require sophisticated analyses.

In contrast, our algorithm is very simple and as it uses random samples, it does not suffer from the limitations arising from choosing explicit points. Moreover, our analysis is also very clean and  directly using the stronger bound in \Cref{thm:main_high_dim} (instead of Koksma-Hlawka) gives the following improved result.

\begin{restatable}[Best of MC and $\sigma_{\SO}$]{theorem}{BestofBothWorldsSO} \label{thm:best-of-both-2}
For any function $f: [0,1]^d \rightarrow \R$, its integration error using $A_T$ is unbiased, i.e.,~  $\E[\err(A_T,f)] = 0$,  and satisfies that 
\[
\E[(\err(A_T,f))^2] \leq \widetilde{O}_d(1) \cdot \min_{f = g + h} \Big( \frac{\sigma(g)^2}{n} + \frac{\sigma_{\SO}(h)^2}{n^2} \Big) ,
\]
where the infimum is over all ways of writing $f$ as the sum of two functions $g,h: [0,1]^d \rightarrow \R$. 
\end{restatable}
Using the Fourier expressions for $\sigma(f), \sigma_{\HK}(f)$ and $\sigma_{\SO}(f)$ in \eqref{eq:MC_formula_intro},\eqref{eq:MC_HK_formula_intro}  and \eqref{eq:sigma_SO_formula_intro},  the optimal decomposition of $f$ into $g$ and $h$ becomes immediate (the high frequency terms go to $g$ and  the rest to $h$). We only describe this for Theorem \ref{thm:best-of-both-2}, as it is strictly stronger than Theorem \ref{thm:best-of-both-1}. In particular, Theorem \ref{thm:best-of-both-2} implies the following explicit bound in $1$-d,
\[\E[(\err(A_T,f))^2] \leq \widetilde{O}_d(1) \cdot \frac{1}{n^2} \sum_{k \neq 0} \min\, (|k|,n) \cdot 
\widehat{f}(k)^2.
\]
and more generally for higher dimensions, we have
\[
\E[(\err(A_T,f))^2] \leq \widetilde{O}_d(1) \cdot \frac{1}{n^2} \sum_{\bfk \neq \bf0}  \min \Big(\prod_{j \in [d]} \max(1,|\bfk_j|), n \Big) \cdot \widehat{f}(\bfk)^2.
\]

\subsection{Overview of Ideas}
\label{sec:overview}
We give a high-level overview of the ideas, deferring the technical details for later.
Our starting point is the well-known {\em transference principle} in geometric discrepancy theory, that produces point sets with low continuous discrepancy using tools from {\em combinatorial} discrepancy. We recall this briefly, and refer to \cite{Mat09, AistleitnerBN17} for details.

The idea is to start with a set $A_0 \subseteq [0,1]^d$ of $n^2$ uniformly random points in $[0,1]^d$. Consider the set system with these points as elements, and sets corresponding to all non-empty dyadic rectangles $J = I_1\times \cdots \times I_d$, where each $I_i$ is a dyadic interval of length an integer multiple of $2^{-h}$, where $h=O(\log n)$. 
Using standard arguments (see \Cref{subsec:subg_disc}), this system always has a balanced $\pm 1$ coloring (with equal $+1$ and $-1$) with combinatorial discrepancy $\widetilde{O}_d(1)$. 
Applying this recursively for $T= \log n$ steps gives a partition of $A_0$ into $n$ equal-sized sets, each of which has continuous discrepancy $\widetilde{O}_d(1)$ for each dyadic rectangle. Let us use $A_T$ to denote any such set.

Writing each corner\footnote{Where each $\bfz_i$ is an integer multiple of $2^{-h}$.}  $C_\bfz$ (see \Cref{subsec:geo_disc_num_int})  as a union of $\widetilde{O}_d(1)$ dyadic rectangles, 
this gives $D_2^*(A_T) = \widetilde{O}_d(1)$, which implies an integration error of $\widetilde{O}_d(\sigma_{\HK}(f) / n)$ by the Koksma-Hlawka inequalities.

\smallskip
\noindent{\bf Subgaussian Colorings and Best of Both Worlds.}
Recent progress in combinatorial discrepancy \cite{bdgl18,ALS21,HSSZ24} 
has led to extremely efficient algorithms where the output colorings not only have low discrepancy, but also sufficient {\em randomness}\footnote{This is captured by the notion of subgaussianity which is formally defined in \Cref{subsec:subg_disc}.} (see \Cref{subsec:subg_disc}). 
Using these in the transference principle leads to our main algorithm $\SubgSparsification$ (see \Cref{sec:subg_sparsification_alg}). 

Our first (simple) observation is that this can be used to obtain the best of both worlds results.
Roughly speaking, these properties imply that each $A^{}_T$ behaves like a random subset of $A_0$. As $A_0$ consists of randomly chosen points and satisfies MC bounds, the sets $A_T$  also satisfy this (in addition to having low discrepancy and satisfying QMC bounds).
We remark that similar ideas have also been used in the related but different contexts of designing randomized controlled trials \cite{HSSZ24} and coreset constructions \cite{PT20,CKW24, DM22}.

At first glimpse, a best of both worlds result seems to be as much as one could hope for. 
Indeed, such results are the main conclusions in the other contexts mentioned above. 
We now show how one can go much further, and obtain the bounds in Theorem \ref{thm:main_high_dim}.

\smallskip
\noindent \textbf{Beyond Hardy-Krause.}
The main idea to go beyond the Hardy-Krause variation (and substantially improve over the Koksma-Hlawka inequalities) is to exploit delicate cancellations in the Hlawka-Zaremba formula by combining subgaussian properties of
$\SubgSparsification$ with the combinatorial structures of the discrepancy problems arising in the transference principle.

For simplicity, let us focus on the 1-d case. 
Here, the  Hlawka-Zaremba formula  gives that
\begin{align} \label{eq:Koksma_Hlawka_techniques_2}
\eps \approx \int_0^1 \frac{\disc(x)}{n} \cdot f'(x) d x \approx \frac{1}{n^2} \sum_{j=1}^{n} \disc \big(\frac{j}{n} \big) \cdot f'\big(\frac{j}{n} \big) =: \frac{1}{n^2} \big\langle \bfd^{\mathcal{C}} , \mathbf{f'} \big\rangle ,
\end{align}
where $\disc(x)$ is the combinatorial discrepancy\footnote{Here we have replaced the continuous discrepancy in \eqref{eq:1d-qmc-error} by the combinatorial discrepancy, which is roughly the same in our context by the transference principle.} of the prefix interval $[0,x]$, and $\bfd^{\mathcal{C}}$ is the vector formed by the discrepancy of all prefix intervals of length $j/n$.  As a notational remark, we always use $\disc(\cdot)$ and $\bfd$ to indicate combinatorial discrepancy, and $D(\cdot)$ to denote continuous discrepancy.

As $|\eps| \approx \frac{1}{n^2} |\bfd^{\mathcal{C}} \cdot \mathbf{f'}|$, 
we can perhaps hope to exploit some cancellations in the inner product (notice that the Koksma-Hlawka inequalities give up completely on the possibility of leveraging cancellations --- as they use H\"older's inequalities to bound $\eps$ by $ \|\bfd^{\mathcal{C}}\|_\infty \cdot \|\mathbf{f'}\|_1$ or $ \|\bfd^{\mathcal{C}}\|_2 \cdot \|\mathbf{f'}\|_2$).

\smallskip
\noindent\textbf{A Promising First Idea.}
Let us consider this more closely. As the prefix intervals have bounded discrepancy, we have $\|\bfd^{\mathcal{C}}\|_\infty = O(1)$ and hence Koksma-Hlawka gives that $\|\bfd^{\mathcal{C}}\|_2 \cdot \|\mathbf{f'}\|_2 \approx \sqrt{n} \|\mathbf{f'}\|_2$. Now suppose that we could additionally guarantee that $\bfd^{\mathcal{C}}$ was also $\widetilde{O}(1)$-subgaussian --- then we would actually have  $|\bfd^{\mathcal{C}} \cdot \mathbf{f'}| \approx \|\mathbf{f'}\|_2$, giving a $\sqrt{n}$ factor improvement over Koksma-Hlawka!

In fact,
modern combinatorial discrepancy algorithms can ensure that the output vector $\bfd^{\D}$ consisting of discrepancy of dyadic intervals is $\widetilde{O}(1)$-subgaussian (see \Cref{thm:self_bal_walk}). 
And indeed, as each prefix interval $[0,j/n]$ is a sum of $\widetilde{O}(1)$ dyadic intervals, it seems natural to expect that $\bfd^{\mathcal{C}}$ should also be $\widetilde{O}(1)$-subgaussian.

\smallskip
\noindent\textbf{Failure of Subgaussianity for Prefix Intervals.} 
Unfortunately and perhaps surprisingly, even though each prefix interval comprises of $\widetilde{O}(1)$ dyadic intervals, the subgaussian property is much more delicate and is completely destroyed when we pass from dyadic to prefix intervals.\footnote{\label{foonote:subg_fail}In the 1-d case, one can design tailor-made  algorithms that can guarantee $\widetilde{O}(1)$-subgaussianity of prefix intervals, e.g. by pairing consecutive points and flipping colors, but no such algorithm exists for $d \geq 2$ (see \Cref{sec:example_subg_fail}).}

We give simple but (very) instructive examples in \Cref{sec:example_subg_fail}, which show that $\bfd^{\mathcal{C}}$ can be $\Omega(n)$-subgaussian, even though $\bfd^{\mathcal{D}}$ is $\widetilde{O}(1)$-subgaussian. 
Roughly, the problem is that prefix intervals overlap a lot, e.g., $[0,1/n]$ is contained in every prefix, which leads to huge dependencies. 
As a result, this rules out any approach based on exploiting the subgaussianity of $\bfd^{\mathcal{C}}$.

\smallskip
\noindent \textbf{Dyadic Decomposition Meets Hlawka-Zaremba: Cancellation to the Rescue.}
Surprisingly, we are able to bypass this problem by leveraging the structure of dyadic decomposition carefully to obtain cancellations of certain Fourier frequencies over long intervals in the Hlawka-Zaremba formula. 
This is our key technical insight and
we sketch it briefly here to give the main idea. But the details are somewhat technical and are given in \Cref{sec:full_proof}.

We start by writing $\bfd^{\mathcal{C}} = P \cdot \bfd^{\D}$, where $P$ is the dyadic decomposition matrix that expresses prefix intervals as combinations of dyadic intervals. 
Then \eqref{eq:Koksma_Hlawka_techniques_2} can be written as 
\[
\eps \approx \frac{1}{n^2} \big\langle \bfd^{\mathcal{C}} , \mathbf{f'} \big\rangle = \frac{1}{n^2} \big\langle \bfd^{\D} , P^\top \mathbf{f'} \big\rangle \approx \frac{1}{n^2} \cdot \big\| P^\top \mathbf{f'} \big\|_2 ,
\]
where the last step crucially uses that $\bfd^{\D}$ is $\widetilde{O}(1)$-subgaussian. So our task reduces to  upper bounding the $\ell_2$ norm $\|P^\top \mathbf{f'}\|_2$. 

However, even though $P$ has $O(\log n)$ row sparsity, its columns can be $\Omega(n)$-dense, e.g., the dyadic interval $[0,1/2]$ is contained in every prefix $[0,j/n]$ for $j \geq n/2$, 
 which means that the column of $P$ corresponding to $[0,1/2]$ has $n/2$ ones. 
Consequently, $P$ has operator norm $\Omega(\sqrt{n})$ -- which
is an equivalent view of the failure of subgaussianity for prefix intervals, as described earlier. 

Nonetheless, we can exploit that the columns of $P$ are highly structured: the ones appearing in each non-zero column of $P$ also correspond, in a natural way, to a dyadic interval $I$, so that the corresponding coordinate of $P^\top \mathbf{f'}$ is simply $\int_I f'(x) = f(r_I) - f(\ell_I)$ (see \Cref{lem:dyadic_decomp_1D}).
Consequently, we can show that the high Fourier frequencies of $f$ exhibit significant cancellations on coordinates of $P^\top \mathbf{f'}$ corresponding to long intervals $I$. Quantifying this leads to our new notion of variation $\sigma_{\SO}$.

There are some additional technicalities to complete the proof of \Cref{thm:main_high_dim}. First, while the Fourier basis is orthogonal, this orthogonality is lost when $P^\top$ is applied and this complicates the calculation. 
To handle this, we apply a random shift to the dyadic system to obtain orthogonality of the Fourier coefficients (see \Cref{lem:orthogonality_fourier}). 
Second, to generalize the 1-d analysis above to higher dimensions, 
we work with the natural {\em tensorization} of 1-d dyadic decomposition (see \Cref{subsec:dyadic_decomp}). However, the higher dimensional version of Hlawka-Zaremba formula has complicated boundary conditions (see \Cref{lem:Hlawka_Zaremba_formula}). This requires exploiting the joint subgaussianity of $\bfd^{\D}$ on these boundaries and performing a similar computation.

\smallskip
\noindent{\bf Organization.}
The rest of the paper is organized as follows. 
We discuss notation and preliminaries in \Cref{sec:prel}, 
and describe the $\SubgSparsification$ algorithm formally in \Cref{sec:subg_sparsification_alg}. In \Cref{sec:best-of-both}, we show how the best of both worlds \Cref{thm:best-of-both-1} follows directly from subgaussianity, and how this idea combined with \Cref{thm:main_high_dim} lead to \Cref{thm:best-of-both-2}.
The proof of \Cref{thm:main_high_dim}
is in \Cref{sec:full_proof}. For ease of exposition, we first give the proof for the 1-d case in \Cref{sec:analysis_1D}, followed by the general case in \Cref{sec:analysis_high_dim}. 
The proof of \Cref{thm:main_high_dim} also suggests a more general definition of $\sigma_{\SO}$, which we discuss in \Cref{subsec:general_def_sigma_SO}. In \Cref{sec:conclude} we give some concluding remarks and describe an intriguing open problem. 
Some further related works and missing details appear in the appendix.




\section{Preliminaries}
\label{sec:prel}
{\bf Notation.}
For integer $j>0$, we denote $[j] :=\{1, \cdots, j\}$. 
We use bold letters, e.g. $\bfz$, to denote vectors, and $\bfz_j$ the $j$th coordinate of $\bfz$. For a subset of coordinates $S \subseteq [d]$, we denote $\overline{S} := [d] \setminus S$ its complement set of coordinates. 
For a vector $\bfz \in \R^d$ and $S \subseteq [d]$, denote $\bfz_{S} \in \R^{|S|}$ the vector $\bfz$ restricted to coordinates in $S$.
For a set of vectors, we use superscript $j$ to denote the $j$th vector $\bfz^j$, to distinguish from the $j$th coordinate. Matrices are denoted by unbolded capital letters. For a matrix $A \in \R^{s \times t}$, we use $A_{j, \cdot}$ to denote its $j$th row and $A_{\cdot, j}$ its $j$th column. 

Given a function $f: [0,1]^d \rightarrow \R$, we denote its mean by $\overline{f} :=\int_{[0,1]^d} f(\textbf{z}) d \textbf{z}$. 
For a finite set of points $A\subseteq [0,1]^d$, the average value of $f$ on $A$ is denoted by $\overline{f}(A) := 
\sum_{\textbf{z} \in A} f(\bfz)/|A|$.
Finally, $\err(A,f) := \overline{f}(A)  - \overline{f}$ denotes the (signed) integration error with respect to point set $A$.

Given a complex number $z \in \C$, we use $z^*$ to denote its complex conjugate, and $|z|$ its modulus.
For a complex vector $\bfz \in \C^d$, we use $|\bfz|$ to denote the $\ell_2$ norm of the vector formed by the modulus of its each coordinate.
For any bounded interval $I \subset \R$, we denote $\ell_I$ and $r_I$ its left and right endpoints. 
Throughout, $\log x$ means the logarithm of $x$ in base $2$.

\subsection{Monte-Carlo and Quasi-Monte-Carlo Methods and Geometric Discrepancy}
\label{subsec:geo_disc_num_int}

In Monte Carlo methods, $A$ is chosen by picking each point independently and uniformly, and the error satisfies 
$\E_A [ \err(A,f)]=0$ and
$\E_A \big[\err(A,f)^2 \big] =\sigma(f)^2/|A|$.
where $\sigma(f)^2$ is the variance of $f$.

Quasi-Monte Carlo (QMC) methods use more structured set of points $A$.
There is a remarkable (exact) formula, due to
Hlawka and Zaremba \cite{Hla61,Zar68},
for $\err(A,f)$
for any set\footnote{The half-open cube here is standard in QMC literature, since it  partitions easily into half-open subcubes
and hyperrectangles. It also does not make any difference for our algorithm as it uses random samples.} $A \subseteq (0,1]^d$ and any function $f$, in terms of the continuous discrepancy of $A$ and derivatives of $f$ (\Cref{lem:Hlawka_Zaremba_formula}). 

For a set of points $A$, let us define the continuous discrepancy
of a measurable set $R \subseteq [0,1]^d$ as \[D(A, R) := |A| \cdot \vol(R) - |A \cap R|.\]
For a family  $\mathcal{R}$ of subsets of $[0,1]^d$, 
we
denote $D(P, \mathcal{R}) := \sup_{R \in \mathcal{R}} |D(P, R)|$.

For a point $\bfx \in [0,1]^d$, the {\em corner} $C_\bfx$ is the set $(0,\bfx_1]\times \cdots \times (0,\bfx_d]$.
We will be interested in discrepancy of corners.
For a function $f: \R^d \rightarrow \C$, let $\partial_{[d]} f(\bfz) := \nicefrac{\partial^d f(\bfz)}{\partial \bfz_1 \cdots \partial \bfz_d}$ denote the mixed partial derivative of $f$. Similarly, for $S \subseteq [d]$, let 
$\partial_S f(\bfz) := \nicefrac{\partial^{|S|} f(\bfz)}{\partial \bfz_{S}}$ denote the mixed partial derivative of $f$ w.r.t. coordinates in $S$. We have the following formula.
\begin{lemma}[Hlawka–Zaremba Formula, \cite{Hla61,Zar68}] \label{lem:Hlawka_Zaremba_formula}
Let $f: [0,1]^d \rightarrow \R$ have continuous mixed derivative $\partial_{[d]} f(\bfz)$. Then for any set of points $A \subseteq (0,1]^d$, 
\begin{align} \label{eq:Hlawka_Zaremba_formula}
\err(A,f) = -\sum_{\emptyset \neq S \subseteq [d]} (-1)^{|S|} \int_{Q_S} h(\bfz) \cdot \partial_S f(\bfz) d \bfz ,
\end{align}
where $h(\bfz) = D(A,C_{\bfz})/|A|$ is the (scaled) continuous discrepancy of the corner $C_{\bfz}$, and $Q_S := \{\bfz \in [0,1]^d: \bfz_j = 1 \text{ for all } j \notin S\}$ denotes the (shifted) $|S|$-dimensional cube in $\R^d$. 
\end{lemma}
The case of $d=1$ is much easier to parse, where $\err(A,f)  = \int_x h(x) f'(x) dx$.
For intuition, we give the derivation of this case, which is based on a simple integration by parts, in \Cref{sec:appendix_QMC}.

Applying Cauchy-Schwartz in \eqref{eq:Hlawka_Zaremba_formula} twice gives the following $\ell_2$-Koksma-Hlawka inequality.

\begin{theorem}[$\ell_2$-Koksma-Hlawka Inequality, \cite{Kok42,Hla61,Zar68}] 
\label{thm:Koksma_Hlawka_ineq}
Let $f: [0,1]^d \rightarrow \R$ have continuous mixed derivative $\partial_{[d]} f(\bfz)$.\footnote{This requirement can be further relaxed. See \cite{Mat09} and the references therein.} Then for any set of points $A \subseteq (0,1]^d$, 
\begin{align} \label{eq:Koksma_Hlawka_ineq} 
| \err(A,f) | \leq h_{2,\mathsf{Proj}}(A, \mathcal{C}^d) \cdot \sigma_{\HK}(f) ,
\end{align}
where $h_{2,\mathsf{Proj}}(A, \mathcal{C}^d) := \big(\sum_{\emptyset \neq S \subseteq [d]} \int_{Q_S} h(\bfz)^2 d \bfz\big)^{1/2}$ is the (scaled) projected $\ell_2$-continuous discrepancy of corners $\mathcal{C}^d$, and $\sigma_{\HK}(f) := \big(\sum_{\emptyset \neq S \subseteq [d]} \int_{Q_S} (\partial_S f(\bfz))^2 d \bfz\big)^{1/2}$ is the $\ell_2$-Hardy-Krause variation. 
\end{theorem}

One may also use the $\ell_1$/$\ell_\infty$-H\"older's inequality in \eqref{eq:Hlawka_Zaremba_formula} to derive an $\ell_1/\ell_\infty$ version of Koksma-Hlawka inequality. We omit the details and refer to \cite{WikiLowDisc,KN12,Owe13} and the references therein.

Since the seminal work of van der Corput \cite{vdC35a,vdC35b}, there have been numerous ingenious constructions of sets with small continuous discrepancy, with the best-known bound being $\widetilde{O}_d(1)$ for all versions of continuous discrepancy. These constructions are typically very delicate, and we discuss them more in \Cref{subsec:related_work}. Some excellent references are \cite{Nie92, DP10, Owe13, Mat09}.

\subsection{Fourier Analysis on the Unit Cube}
\label{subsec:Fourier_analysis}
Let $f: [0,1]^d \rightarrow \C$ be a  $1$-periodic function in $L^2([0,1]^d)$, i.e., $f(\mathbf{0}_j, \bfz_{[d]\setminus \{j\}}) = f(\ind_j, \bfz_{[d]\setminus \{j\}})$ for all $j\in [d]$ and $\bfz \in [0,1]^d$,
and square integrable satisfying $\int_{[0,1]^d} |f(\bfz)|^2 d \bfz < \infty$.

For any $\bfk \in \Z^d$, the Fourier coefficients of $f$ are defined as
$
\widehat{f}(\bfk) := \int_{[0,1]^d} f(\bfz) \exp(-2 \pi i \langle \bfk, \bfz \rangle) d \bfz,$
and the Fourier series of $f$ is given by 
\[ \sum_{\bfk \in \Z^d}\widehat{f}(\bfk) \cdot \exp({2 \pi i \langle \bfk, \bfz \rangle}).\]
For any square-integrable function $f$, its Fourier series converges pointwise to $f$ almost everywhere \cite{carleson66, Fefferman}. This suffices for our purposes as our algorithm uses random samples.
It is also known if $f$ has bounded variation, then this convergence is pointwise everywhere, see e.g., \cite{Grafakos09}.

The orthogonality of the Fourier basis implies the Parseval identity $
\int_{[0,1]^d} |f(\bfz)|^2 d \bfz = \sum_{\bfk \in \Z^d} |\widehat{f}(\bfk)|^2$.

As $\overline{f}= \int_{[0,1]^d} f(\bfz) d\bfz = \widehat{f}(\bf0)$, this gives that  $\sigma(f)^2 = \sum_{\bfk \in \Z^d \setminus \{\mathbf{0}\}} |\widehat{f}(\bfk)|^2$. 
Also note that 
\[\partial_S f(\bfz) = \sum_{\bfk \in \Z^d 
\setminus \{\mathbf{0}\}} \widehat{f}(\bfk) \cdot \exp({2 \pi i \langle \bfk, \bfz \rangle}) \cdot (2 \pi i)^{|S|} \prod_{j \in S} \bfk_j,\] so the $\ell_2$-Hardy-Krause variation is given by
\begin{align*} 
\sigma_{\HK}(f)^2 =  \Theta_d(1) \sum_{\bfk \in \Z^d \setminus \{ \mathbf{0}\}} 
\big|\widehat{f}(\bfk) \big|^2 \cdot \prod_{j \in [d]} \max(1,|\bfk_j|^2).
\end{align*}
Let us also recall our definition of $\sigma_{\SO}(f)$ in \eqref{eq:sigma_SO_formula_intro} below.
\[
\sigma_{\SO}(f)^2 := 
\sum_{\bfk\in \Z^d \setminus \{ \mathbf{0}\}} 
 \big|\widehat{f}(\bfk) \big|^2 \cdot \prod_{j \in [d]} \max(1,|\bfk_j|).
\]

\subsection{Combinatorial Discrepancy and Vector Balancing}
\label{subsec:subg_disc}
In the vector balancing problem, we are given vectors $\bfv^1, \cdots, \bfv^n \in \R^m$ of $\ell_2$ norm at most $1$, and our goal is to find a coloring $\bfx \in \{\pm 1\}^n$ to minimize the {\em combinatorial} discrepancy $\|\sum_{j=1}^n \bfx_j \bfv^j\|_\infty$. 

In seminal work, Banaszczyk \cite{b98} showed that there always exists a coloring $\bfx \in \{\pm 1\}^n$ with discrepancy at most $O(\sqrt{\log (mn)})$. 
Bansal et al.~\cite{bdgl18} gave an efficient algorithm, called the \textsf{Gram-Schmidt Walk}, matching Banaszczyk's bound. In particular, the random coloring $\bfx \in \{\pm 1\}^n$ generated by their algorithm is symmetric\footnote{We call a random vector symmetric if its distribution is symmetric around the origin.} and satisfies that the discrepancy vector $\sum_{j=1}^n \bfx_j \bfv^j$ is $O(1)$-subgaussian, where we recall the definition of subgaussian vectors below. 

\begin{definition}[Subgaussian Vectors] 
A random vector $\bfu \in \R^m$ is called $\sigma^2$-subgaussian if for all $\bfz \in \R^m$, one has $\E[\exp(\langle \bfu, \bfz\rangle)] \leq \exp(\sigma^2 \|\bfz\|_2^2/2)$. 
\end{definition}
The subgaussianity constant 
in the \textsf{Gram-Schmidt Walk} was further improved to $1$ in \cite{HSSZ24}. 
However, the \textsf{Gram-Schmidt Walk} algorithm has a rather large $\poly(n,m)$ runtime.
Recently, \cite{ALS21} gave an elegant online algorithm, called the \textsf{Self-Balancing Walk}, which is extremely fast and has the following (only slightly worse) discrepancy guarantee. 
\begin{theorem}[\cite{ALS21}] \label{thm:self_bal_walk}
Given vectors $\bfv^1, \cdots, \bfv^n \in \R^m$ with $\|\bfv^j\|_2 \leq 1$ online, the \textsf{Self-Balancing Walk} computes a symmetric random coloring $\bfx \in \{\pm 1\}^n$, 
such that $\sum_{j =1}^t \bfx_j \bfv^j$ is $O(\log (mn))$-subgaussian for all $t \in [n]$ whp. 
Moreover, the algorithm runs in $O(\sum_{j \in [n]} \nnz(\bfv^j))$ time.\footnote{$\nnz(\bfz)$ denotes  number of non-zero coordinates in a vector $\bfz$.} 
\end{theorem}

For our purposes, it will be convenient to have $\bfx$ be {\em balanced}, i.e., have equal number of $1$'s and $-1$'s. For even $n$, this can be achieved simply by pairing the adjacent vectors to form the vectors $v_1 - v_2, v_3 - v_4, \cdots, v_{n-1} - v_n$ (note that their length is at most $2$) and applying  \Cref{thm:self_bal_walk}.
We call the resulting modified algorithm\footnote{This modification loses that the algorithm is fully online, and needs a lookahead of one step. 
This does not matter for our application to numerical integration.  
In cases where an online implementation is really required, one can also use \textsf{Self-Balancing Walk} in place of $\SubgVecBal$ in our $\SubgSparsification$ algorithm. 
The resulting integration error is similar, but the analysis is a bit more involved.}   $\SubgVecBal$ and will assume the balanced property henceforth.

\subsection{Dyadic Decomposition}
\label{subsec:dyadic_decomp}

\noindent \textbf{Dyadic Intervals.} Let $j \in \mathbb{N}$, we denote by $\D_j$ the set of all left-open\footnote{Due to a slight technicality caused by the boundary, it is convenient to work with left-open intervals.} dyadic intervals in $(0,1]$ of length $2^{-j}$, i.e. $\D_j := \{(\ell/2^j, (\ell+1)/2^j]: \ell = 0,\cdots, 2^j-1 \}$, and call $j$ the level of these dyadic intervals. For $h \in \mathbb{N}$, denote $\D_{\leq h}$ all dyadic intervals of level $\leq h$, i.e. $\D_{\leq h} := \bigcup_{j \leq h} \D_j$. 
For convenience of our analysis, we also define the right-open dyadic intervals $\widetilde{\D}_{\leq h}$ similarly.

\smallskip
\noindent \textbf{Prefix Intervals and Dyadic Decomposition.} For any $x \in [0,1]$, we denote left-open prefix interval $C_{  x} := (0,x]$, with $C_1 = (0,1]$. For $h \in \mathbb{N}$, we use $\mathcal{C}_h$ to denote all left-open prefix intervals whose length is an integer multiple of $2^{-h}$ (excluding\footnote{This is again due to the slight technicality caused by the boundary.} $(0,1]$), i.e. $\mathcal{C}_h := \{C_{ \ell/2^h}: \ell = 0, \cdots, 2^h-1\}$.

{\em Dyadic decomposition} refers to decomposing sub-intervals of $(0,1]$ into a minimal disjoint union of  dyadic intervals. We only use dyadic decomposition for prefix intervals. 
For any $h \in \mathbb{N}$, we define $P_h \in \{0,1\}^{|\mathcal{C}_h| \times (|\D_{\leq h}|-1)}$ the dyadic decomposition matrix of $\mathcal{C}_h$, where the number of columns is $|\D_{\leq h}|-1$ because we have chosen to exclude the unused $(0,1] \in \D_{\leq h}$ in the decomposition. 
In particular, for the $\ell$th prefix interval $I_\ell \in \mathcal{C}_h$ and the $r$th dyadic interval $J_r \in \D_{\leq h} \setminus \{C_1\}$, we have $(P_h)_{\ell, r} = 1$ if $J_r$ is used in the dyadic decomposition of $I_\ell$, and $0$ otherwise.

As each prefix in $\mathcal{C}_{\leq h}$ is a sum of at most $h$ dyadic intervals in $\D_{\leq h}$, each row of $P_h$ has at most $h$ ones. However, the columns of $P_h$ could be dense with $\Omega(2^h)$ entries. For our analysis, we will in fact crucially exploit the structure of the columns of $P_h$. 
 For technical convenience, we actually work with a slight variant of $P_h$ that has more structured columns, which we call the {\em structured decomposition matrix} and denote it as $\overline{P}_h$. \Cref{lem:aug_trans_mat} below  gives its crucial properties.
We postpone the definition of $\overline{P}_h$ and the proof of \Cref{lem:aug_trans_mat} to \Cref{sec:missing_proofs}.

\begin{restatable}[Properties of Structured Decomposition Matrix]{lemma}{AugDecompMatrix} \label{lem:aug_trans_mat}
There exists a structured decomposition matrix $\overline{P}_h \in \{0,1\}^{|\mathcal{C}_h| \times (|\D_{\leq h}|-1)}$ that satisfies the following. 
\begin{enumerate}
    \item \label{property_1_lem:aug_trans_mat} The ones in each column of $\overline{P}_h$ appear consecutively, and moreover these locations form right-open dyadic intervals in $\widetilde{\D}_{\leq h} \setminus \{[0,1)\}$.
    That is, each column can be associated to a right-open dyadic interval $\widetilde{I} \in \widetilde{\D}_{\leq h} \setminus \{[0,1)\}$, and the ones in that column are exactly at rows $C_z$ with $z \in 2^{-h} \Z \cap \widetilde{I}$. In particular, for every $1 \leq \ell \leq h$, there are exactly $2^\ell$ columns of $\overline{P}_h$ with exactly $2^{h-\ell}$ ones.
    
    \item \label{property_2_lem:aug_trans_mat} For every vector $\bfu \in \R^{|\mathcal{C}_h|}$, we have $\|P_h^\top \bfu\|_2 \leq \|\overline{P}_h^\top \bfu\|_2$. 
\end{enumerate}
\end{restatable}

\smallskip
\noindent \textbf{Higher Dimensions.} We also use the higher dimensional versions of the notions above, many of which are simply tensor products of their one dimensional counterpart. 
In particular, $\D_{\leq h}^{\otimes d} := \{I_1 \times \cdots \times I_d: I_j \in \D_{\leq h} \text{ for all $j \in [d]$}\}$ denotes all dyadic boxes with level at most $h$ in each dimension. 
For any $\bfz \in [0,1]^d$, recall from \Cref{subsec:geo_disc_num_int} that we denote $C_{\bfz} := (0,\bfz_1] \times \cdots \times (0,\bfz_d]$ the left-open corner at $\bfz$. 
We use $\mathcal{C}_h^{\otimes d} := \{I_1 \times \cdots \times I_d: I_j \in \mathcal{C}_h \text{ for all $i \in [d]$}\}$ to denote all left-open corners at grid points in $2^{-h} \Z^d \cap [0,1)^d$. 
The dyadic decomposition of $\mathcal{C}_h^{\otimes d}$ into dyadic boxes $\D_{\leq h}^{\otimes d}$ is given by $P_h^{\otimes d}$. 
We also denote $\mathcal{C}^d := \{C_{\bfz}: \bfz \in [0,1]^d\}$ the set of all left-open corners.

\section{The Subgaussian Transference Algorithm}
\label{sec:subg_sparsification_alg}

We formally describe the $\SubgSparsification$ algorithm and note some simple properties.  

\subsection{The Algorithm}

Set $h = O(\log (dn))$ (with a large enough constant). $\SubgSparsification$ begins by applying a random shift to the dyadic system $\D_{\leq h}^{\otimes d}$, i.e. the dyadic boxes are built on the (folded) box $(\bfs, \ind + \bfs]$ (\textsf{mod} $[0,1)^d$) for uniformly random $\bfs \sim [0,1)^d$. Here, any coordinate exceeding $1$ is folded back to $[0,1)$ in every dimension. To simplify notation, we only describe the algorithm for the shift $\bfs = \mathbf{0}$. 

$\SubgSparsification$ starts with a set $A_0$ consisting of $n_0 = n^2$ independent random samples from $[0,1]^d$.
We assume wlog that $n$ is an integer power of $2$.
Let us denote $A_0$ by $A_0^{(0)}$.

Let $T:= \log n$. At each step $0\leq t<T$, we have $2^t$ sets $A_t^{(i)}$ for $0 \leq i < 2^t$, and
each set $A_t^{(i)}$  is split into two equal-sized sets $A_{t+1}^{(2i)}$ and $A_{t+1}^{(2i+1)}$ as described below.
For ease of notation, let $A_t$ denote a generic set $A_t^{(i)}$.

To split $A_t$, 
we run the $\SubgVecBal$ algorithm from \Cref{subsec:subg_disc} on the following set of vectors. For the $j$th point $\bfz^j \in A_t$, define $\bfv^j$ to be the stacking\footnote{To the best of our knowledge, the idea of stacking vectors to ensure subgaussianity of the colorings first appeared in \cite{HSSZ24}, and was also crucially exploited in other works (e.g., \cite{BJM+22}).} of (1) the standard basis vector $e^{A_t}_j \in \{0,1\}^{|A_t|}$, with a single one in the $j$th coordinate, and (2) the incidence vector of $\bfz^j$ w.r.t. the dyadic boxes  $\D_{\leq h}^{\otimes d}$. That is,
\begin{align} \label{eq:stacked_vector}
\bfv^j = \Big( e^{A_t}_j , \big(\ind_{\{\bfz^j \in B\}}\big)_{B \in \D_{\leq h}^{\otimes d}} \Big) ,
\end{align}
where $\ind_{\{\bfz^j \in B\}}$ is $1$ if $\bfz^j \in B$ and $0$ otherwise. 

Let $\bfx^t \in \{\pm 1\}^{|A_t|}$ denote the (balanced) coloring produced by $\SubgVecBal$.
We use this to split $A_t$ into the sets $A_{t+1} := \{\bfz^j \in A_t: \bfx^t_j = -1\}$ and $A'_{t+1} := A_t \setminus A_{t+1} = \{\bfz^j \in A_t: \bfx^t_j = +1\}$. 
As $\bfx$ is balanced, $|A_{t+1}| = |A_t|/2$, and consequently the final sets $A_T$ satisfy $|A_T| = n$. 

\begin{remark}
By symmetry of $\SubgVecBal$,  
the  $2^T$ sets $A_T^{(0)}, \cdots, A_T^{(2^T-1)}$ have the same distribution, and so the superscripts do not matter while discussing the properties of these sets below. For concreteness, we use $A_t$ to refer to $A_t^{(0)}$. Also notice that the algorithm does not depend in any way on the function $f$ to be integrated.
\end{remark}

\subsection{Properties of $\SubgSparsification$}
\label{subsec:property_subg_sparse}

\noindent \textbf{Runtime.} $\SubgSparsification$ has overall runtime $\widetilde{O}_d(n^2)$, and hence $\widetilde{O}_d(1)$ amortized time per output point. This is because at each level $t$ of the algorithm, the vectors $\bfv^j$ are all $\widetilde{O}_d(1)$-sparse, and there are $n^2$ such vectors overall in the $2^t$ sets $A_t^{(0)}, \cdots, A_t^{(2^t-1)}$. By \Cref{thm:self_bal_walk}, the algorithm takes $\widetilde{O}_d(n^2)$ time in each level, the total runtime bounds  follows as there are $T = \log n$ levels. 

\begin{remark}[Implicit Representation of Vectors]
Strictly speaking, the vectors $\bfv^j$ actually have dimension $n^{O(d)}$. But since they are all $\widetilde{O}_d(1)$-sparse, the algorithm represents these vectors implicitly by only recording their non-zero coordinates. 
\end{remark}

\smallskip
\noindent \textbf{Subgaussianity.}
Let us consider an iteration $t \in \{0,\cdots, T-1\}$ of the algorithm where $A_{t+1}$ is produced from $A_t$ by running $\SubgVecBal$.
For any set $C \subseteq [0,1]^d$, let
\begin{align} \label{eq:comb_disc_defn}
\disc_t(C) &:=  | \{\bfz^j \in C \cap A_t \text{ with } \bfx^t_j = +1\}| - | \{\bfz^j \in C \cap A_t \text{ with } \bfx^t_j = -1\}| \nonumber \\
& = |C \cap A_t| - 2 |C \cap A_{t+1}|. 
\end{align}
denote the combinatorial discrepancy of $C \cap A_t$. In particular, we denote $\disc_t(\bfz) := \disc_t(C_{\bfz})$ for any $\bfz \in [0,1]^d$. 
Our algorithm guarantees subgaussianity of the coloring $\bfx^t$ as well as the combinatorial discrepancy for all left-open dyadic boxes $C \in \D_{\leq h}^{\otimes d}$.  
Let us denote this vector as $\bfd^{t,\D} \in \Z^{|\D_{\leq h}^{\otimes d}|-1}$ with coordinate $\bfd^{t,\D}_B := \disc_t(B)$ for any left-open dyadic box $B \in \D_{\leq h}^{\otimes d} \setminus C_{\ind}$.

\begin{lemma}[Subgaussianity for Dyadic Boxes] \label{lem:subg_dyadic}
For any $0 \leq t \leq T-1$, $\SubgSparsification$ satisfies (1) $\disc_t([0,1]^d) = 0$, and (2) the random vector $(\bfx^t, \bfd^{t,\D})$ is $O_d(\log^{d+1} n)$-subgaussian. 
\end{lemma}
\begin{proof}
The first claim follows because $\SubgVecBal$ always produces a balanced coloring. 

Next, for every point $\bfz^j \in A_t$, the vector $\bfv^j$ is $(h+1)^d$-sparse. This follows as for any fixed choice of $h_1,\cdots,h_d$ with $0\leq h_i\leq h$,  the point $\bfz^j$ lies in exactly one dyadic interval of size $2^{-h_1}\times \cdots \times 2^{-h_d}$.
 Thus $\|\bfv^j\|_2^2 = O_d(\log^d n)$, and the claim follows from Theorem \ref{thm:self_bal_walk} (by scaling) as we run $\SubgVecBal$ on the at most $n^2$ vectors $\bfv^j$. 
\end{proof}

\subsection{A Formula for Continuous Discrepancy}
Recall that in \Cref{lem:Hlawka_Zaremba_formula}, the integration error of a point set $A$ depends on the (scaled) continuous discrepancy of $A$ w.r.t. corners, i.e. $h(\bfz) = D(A,C_{\bfz})/|A|$.  We now give an expression for it.

For each $t \in \{0, \cdots, T\}$, we denote $n_t := |A_t| = 2^{-t}n_0$. 
For any measurable set $C \subseteq [0,1]^d$, 
denote $n_t(C) := |C \cap A_t|$, and $h_t(C) := D(A_t,C)/n_t = \vol(C) - n_t(C)/n_t$.
\begin{lemma}
\label{lem:decomp_cont_disc}
For any measurable set $C \subseteq [0,1]^d$, we have
$h_T(C) 
 = h_0(C) + \sum_{t=0}^{T-1} \disc_t(C)/n_t$.
\end{lemma}
\begin{proof}
As $h_T(C) = \vol(C) -n_T(C)/n_T$ and $\vol(C) = h_0(C) + n_0(C)/n_0$, we can write $h_T(C)$ as the telescoping sum
\[ 
h_T(C) =  
h_0(C) + \sum_{t=0}^{T-1} \Big(\frac{n_t(C)}{n_t} - \frac{n_{t+1}(C)}{n_{t+1}}\Big) .\]
As $n_t(C) = |C \cap A_t|$ and $n_t= 2 n_{t+1}$ and using  \eqref{eq:comb_disc_defn}, each summand above is exactly
\begin{align*} 
\frac{n_t(C)}{n_t} - \frac{n_{t+1}(C)}{n_{t+1}} 
& = \frac{|C \cap A_t|}{n_t} - \frac{2 |C \cap A_{t+1}|}{n_t} = \frac{\disc_t(C)}{n_t}. \qquad \qedhere 
\end{align*}
\end{proof}

\section{Achieving Best of Both Worlds}
\label{sec:best-of-both}

In this section, we prove that the numerical integration error of $\SubgSparsification$ is, up to logarithmic factors, at least as good as the better of Monte Carlo and quasi-Monte Carlo methods.

\BestofBothWorlds*

We break the proof of  \Cref{thm:best-of-both-1} into three parts below. Exactly the same argument, but using Theorem \ref{thm:main_high_dim} instead of Lemma \ref{lem:better_than_MC},  gives the improved bound in Theorem \ref{thm:best-of-both-2}, restated below.
\BestofBothWorldsSO*

\begin{lemma}[Unbiasedness] \label{lem:unbiasedness}
For each $t=0,\cdots,T$, we have $\E[\err(A_t,f)] = 0$. 
\end{lemma}
\begin{proof}
As each sample $\bfz^j \in A_0$ is picked uniformly in $[0,1]^d$, we have $\E_{A_0} [f(\bfz^j)]=\overline{f}$, and the claim holds for $t=0$.
For $t\geq 1$ we have,
\[  
\E[\overline{f}(A_t)] = \frac{1}{n_t}  \E\Big[ \sum_{\bfz \in A_t} f(\bfz) \Big] = \frac{1}{n_t}\E_{A_0}\Big[\sum_{\bfz \in A_0}  f(\bfz) \cdot \E\big[\ind_{\{\bfz \in A^t\}}\big] \Big], 
\]
where the inner expectation is over the randomness of $\SubgVecBal$.  

As $\SubgVecBal$ outputs a symmetric distribution over colorings, each $\bfz \in A_{t-1}$ lies in $A_{t}$ with probability exactly $1/2$. So $\E\big[\ind_{\{\bfz \in A^t\}}\big] = 2^{-t} = n_t/n_0$ and the RHS
is  $\frac{1}{n_0} \E_{A_0}[\sum_{\bfz \in A_0} f(\bfz) ] =  \overline{f}$.
\end{proof}

We now prove that $\SubgSparsification$ has no worse error than quasi-Monte Carlo. 
This is essentially the folklore proof of the transference principle, and we give it here for completeness.

\begin{lemma}[$\SubgSparsification$ and QMC] \label{lem:better_than_QMC}
With high probability, 
 $|\err(A_T,f)| = \widetilde{O}_d (V_{\HK}(f)/n)$. 
\end{lemma}
\begin{proof}
By the Koksma-Hlawka inequality \eqref{eq:Koksma_Hlawka_ineq}, it suffices to show $D(A_T, \mathcal{C}^d) \leq \widetilde{O}_d(1)$ whp. 

To this end, we first prove that $D(A_T, \D_{\leq h}^d) \leq \widetilde{O}_d(1)$ whp. Fix a dyadic box $B \in \D_{\leq h}^d$.  By \Cref{lem:decomp_cont_disc} 
we have
\begin{align}
\label{eq:qmc-rel}
D(A_T, B) = n_T \cdot h_T(B)  = n_T \cdot h_0(B) + n_T \cdot \sum_{t=0}^{T-1} \frac{\disc_t(B)}{n_t}
 = \frac{D(A_0, B)}{n} + \frac{1}{n} \sum_{t=0}^{T-1} 2^t \cdot \disc_t(B) .
\end{align}
where the last equality uses that $n_T=n, n_0=n^2$ and $n_t = n_0/2^t$.

We note that $D(A_0, B) = \sum_{\bfz \in A_0} \big(\E[\ind_{\{\bfz \in B\}}] - \ind_{\{\bfz \in B\}}\big)$, where $\ind_{\{\bfz \in B\}} \in \{0,1\}$ are i.i.d.~random variables. Consequently, by Chernoff bound and a union bound over the $O_d(n^d)$ dyadic boxes $B$, we have that $D(A_0, B)  \leq O( (n_0 \log n^d)^{1/2}) = O_d(n \log^{1/2} n)$ whp for all $B \in \D_{\leq h}^{\otimes d}$. So the first term $D(A_0, B)/n$ on the right side of \eqref{eq:qmc-rel} is $\widetilde{O}_d(1)$. 

Similarly, the second term on the right is $\widetilde{O}_d(1)$ as $2^T/n=1$ and
by \Cref{lem:subg_dyadic} the vector $\bfd^{t,\mathcal{D}}$ is  $O(\log^{d+1} n)$-subgaussian, and hence in particular each coordinate $\disc_t(B)$ is $O(\log^{d+1} n)$ subgaussian. Thus $D(A_T, B) \leq O_d(\log^{d/2 + 1} n)$ whp for all $B \in \D_{\leq h}^{\otimes d}$. 

As each corner $C \in \mathcal{C}_h^{\otimes d}$ can be decomposed as at most $O_d(\log^d n)$ dyadic boxes, by subgaussianity we also have that  $D(A_T, C) \leq    O_d(\log^{d + 1} n) = \widetilde{O}_d(1)$ whp for all $C \in \mathcal{C}_h^{\otimes d}$.
Finally, since the initial set of samples $A_0$ is uniformly random, the number of samples in every strip of width $2^{-h}$, i.e. strips in $\D_h \times [0,1]^{\otimes (d-1)}$ is at most $\widetilde{O}_d(1)$. Consequently, each corner $C \in \mathcal{C}^d$ (recall that this is the infinite set of all possible corners) has the same discrepancy as some corner in $\mathcal{C}_h^{\otimes d}$ up to an additive $\widetilde{O}_d(1)$ term. 
\end{proof}

Finally, $\SubgSparsification$ has error at most that of Monte Carlo (up to $\widetilde{O}_d(1)$ factors)\footnote{This factor can be reduced to $O(\log n)$ if one cares, by replacing each $e_j^{A_t}$ in \eqref{eq:stacked_vector} by $(h+1)^{d/2} e_j^{A_t}$. By rescaling, it is easily verified that the coloring $\bfx^t$ produced by the ALS algorithm will be $O(\log n)$-subgaussian.}.

\begin{lemma}[$\SubgSparsification$ and Monte Carlo] \label{lem:better_than_MC}
The error satisfies $\E[\err(A_T,f)^2] \leq \widetilde{O}_d \Big( \frac{\sigma(f)^2}{n} \Big).
$
\end{lemma}
\begin{proof}
As $
\err(A_{t},f)  = \overline{f}(A_{t})  - \overline{f}$  by definition, we can recursively write
$\err(A_{t+1},f) 
 = \overline{f}(A_{t+1}) -  \overline{f} =  \err(A_{t},f)  - \Delta_t$,
where 
\[  \Delta_t :=   \overline{f}(A_{t}) - \overline{f}(A_{t+1})   =  
\frac{1}{n_t} \sum_{\bfz \in A_t} f(\bfz) - \frac{1}{n_{t+1}} \sum_{\bfz \in A_{t+1}} f(\bfz)  = \frac{1}{n_t} \sum_{\bfz^j \in A_t} \bfx^t_j \cdot f(\bfz^j),\]
where we use that $n_{t+1}=n_t/2$ and $A_{t+1}$ is obtained from $A_t$ by keeping the $\bfz^j$ for which $\bfx^t_j=-1$.
\allowdisplaybreaks
Conditioned on $A_t$, 
as the coloring $\bfx^t$ produced by $\SubgVecBal$ is symmetric, $\E[ \Delta_t | A_t]= 0$. So,
\[\E \big[\err(A_{t+1},f)^2 \big] 
 = \E\Big[\E\big[\err(A_{t+1},f)^2 \big | A_t \big] \Big] \\
 = \E\Big[\err(A_t,f)^2 +  \E \big[\Delta_t^2 \big | A_t \big] \Big].\]
As $\bfx^t \in \{\pm 1\}^n$ is $O(\log^{d+1} n)$ subgaussian in \Cref{lem:subg_dyadic}, we have 
\[\E\Big[ \big(\sum_{j\in A_t} \bfx^t_j f(z^j) \big)^2 \Big] = O(\log^{d+1} n) \sum_{\bfz^j \in A_t} f(\bfz^j)^2\] and thus,
\begin{align*}
\E\big[\err(A_{t+1},f)^2 \big]-  \E\big[\err(A_t,f)^2 \big ] 
&  = 
\E\big[\E[ \Delta_t^2 |A_t]\big]   =\E \Big[ \frac{O(\log^{d+1} n)}{n_t^2} \cdot \sum_{\bfz^j \in A_t} f(\bfz^j)^2 \Big] \\
& =   \frac{O(\log^{d+1} n)}{n_t^2} \cdot \E\Big[ \sum_{\bfz^j \in A_0} f(\bfz^j)^2 \cdot \ind_{\{\bfz^j \in A_t\}} \Big] \\
& =  \frac{O(\log^{d+1} n)}{n_t n_0} \cdot \E\Big[ \sum_{\bfz^j \in A_0} f(\bfz^j)^2\Big] \leq  \frac{O(\log^{d+1} n)}{n_t} \cdot \sigma(f)^2.
\end{align*}

Summing the above over $t = 0, \cdots, T$, we obtain
\begin{align*}
\E \big[\err(A_T,f)^2 \big] = \E \big[\err(A_0,f)^2 \big] + \frac{O(\log^{d+1} n)}{n} \cdot \sigma(f)^2 .
\end{align*}
The lemma now follows as $\E \big[\err(A_0,f)^2 \big] = \sigma(f)^2/n_0 = \sigma(f)^2/n^2 $ and hence negligible. 
\end{proof}

Now we are ready to prove \Cref{thm:best-of-both-1} from the three lemmas above. 

\begin{proof}[Proof of \Cref{thm:best-of-both-1}]
The first statement of the theorem is given by \Cref{lem:unbiasedness}. 
For any $g,h: [0,1]^d \rightarrow \R$ for which $f = g + h$, note that $\err(A_T,f) = \err(A_T,g) + \err(A_T,h)$. Then the theorem follows immediately from \Cref{lem:better_than_QMC,lem:better_than_MC}. 
\end{proof}

\section{Beyond Hardy-Krause Variation}
\label{sec:full_proof}

In this section, we present the formal proof of  \Cref{thm:main_high_dim} which is restated below. 

\IntegralSubgHighDim*

\subsection{Proof for 1-D}
\label{sec:analysis_1D}

We start with the proof of \Cref{thm:main_high_dim} for $d=1$. This already contains most of the ideas, and the proof for higher dimensions will use several key lemmas from the $1$-d analysis.

Let $f \in L^2([0,1])$ be an arbitrary but fixed 1-periodic function with a continuous derivative.\footnote{These assumptions can be relaxed. See Remarks \ref{remark:smoothness} and \ref{remark:general_form} in \Cref{subsubsec:beyond_hardy-krause}.} As $\SubgSparsification$ applies a uniformly random $s \sim [0,1)$ to the dyadic system $\D_{\leq h}$, we may equivalently view it as keeping $\D_{\leq h}$ fixed but applying a random shift to $f$ and work with $f_s(x) := f((x + s) \textsf{ mod } [0,1))$, for a uniformly random $s \sim [0,1)$.  
We use $\E_s$ to denote expectation w.r.t. the random shift $s$. 
For technical convenience, we equivalently analyze the error $\err(A_T, f_s)$, where the effect of the random shift $s$ is completely subsumed in $f_s$ (i.e., $A_T$ does not depend on $s$).

In 1-d, the Hlawka-Zaremba identity in \Cref{lem:Hlawka_Zaremba_formula} gives that the integration error is \[
\err(A_T,f_s) = \int_0^1 h_T(z) \cdot f_s'(z) d z ,\]
where $h_t(z) = D(A_t, C_z)/n_t  = h_0(z) + \sum_{t=0}^{T-1} \disc_t(z)/n_t$ by \Cref{lem:decomp_cont_disc}. Thus we have,
\begin{equation} \label{eq:error_decomp_1D} 
\err(A_T,f_s)  = \err(A_0,f_s) + \sum_{t=0}^{T-1} \err^{\disc}_t(f_s) =: \err(A_0,f_s) +  \err^{\disc}(f_s) ,
\end{equation}
where $\err^{\disc}_t(f_s) :=\frac{1}{n_t}  \int_0^1 \disc_t(z)f_s'(z) d z$ denotes the {\em discrepancy error} at step $t$.

\smallskip
\noindent \textbf{Bounding the Discrepancy Error.} 
Define the vector $\bfu^{f_s} \in \R^{|\mathcal{C}_h|}$ with coordinates  
\begin{align} \label{eq:u_f_vec}
\bfu^{f_s}_j := \int_{j/2^h}^{(j+1) /2^h} f_s'(x) d x = f_s((j+1) /2^h) - f_s(j /2^h) 
\end{align}
for $j \in \{0, \ldots, 2^h-1\}$. 

Note that the combinatorial discrepancy for prefixes $\disc_t(z)$ (in the formula for $\err_t^{\disc}$), viewed as a function in $z \in [0,1]$, is a step function that changes its value only at points in $A_t$. 
For a cleaner presentation of our analysis, we make the simplifying assumption that $\disc_t(z)$ changes values only at the {\em grid points} in $2^{-h} \Z \cap (0,1]$. This assumption is without loss of generality\footnote{\label{footnote:perturb_disc}This assumption can be removed by carefully defining a perturbation of $\disc_t(z)$ as a convex combination of $\disc_t(j/2^h)$ and $\disc_t((j+1)/2^h)$, where $z \in (j/2^h, (j+1)/2^h]$, that satisfies the assumption and \eqref{eq:disc_err_1D}. Using this perturbed discrepancy function, one may proceed with exactly the same analysis that is presented here.} and is only used for the purpose of analysis (but not the algorithm). 
With this assumption, we may view $\disc_t(z)$ as a vector $\bfd^{t, \mathcal{C}} \in \Z^{|\mathcal{C}_h|}$ with coordinates \[\bfd^{t, \mathcal{C}}_j = \disc_t(j/2^h)\] for $j \in \{0, \ldots, 2^h-1\}$.\footnote{Note that we have ignored $\disc_t(1)$ here, since this only corresponds to the single point $z = 1$ with measure $0$ and therefore doesn't contribute to the integral $\int_0^1 \disc_t(z) f_s'(z) d z$.} The superscript $\mathcal{C}$ in this notation is to remind the reader that $\bfd^{t, \mathcal{C}}$ is the combinatorial discrepancy of all left-open prefix intervals in $\mathcal{C}_h$, and should not to be confused with the combinatorial discrepancy vector $\bfd^{t, \mathcal{D}}$ of left-open dyadic boxes defined earlier in \Cref{subsec:property_subg_sparse}. 

Now the discrepancy error $\err^\disc_t$ can be written as an inner product 
\begin{align} \label{eq:disc_err_1D}
\err^\disc_t(f_s) = \frac{1}{n_t} \int_0^1 \disc_t(z) f_s'(z) d z 
& = \frac{1}{n_t} \sum_{j=0}^{2^h-1} \bfd^{t, \mathcal{C}}_j \cdot \bfu^{f_s}_j = \frac{1}{n_t} \cdot \big\langle \bfd^{t, \mathcal{C}} , \bfu^{f_s} \big\rangle . 
\end{align}
To make use of the subgaussianity of  $\bfd^{t, \mathcal{D}}$ in \Cref{lem:subg_dyadic}, we need to express the combinatorial discrepancy $\bfd^{t, \mathcal{C}}$ of prefix intervals in terms of  the dyadic boxes. 
This can be done using the dyadic decomposition matrix $P_h$ defined in \Cref{subsec:dyadic_decomp}:
\[
\bfd^{t, \mathcal{C}} = P_h \cdot \bfd^{t, \mathcal{D}} \quad \text{ and therefore } \quad \frac{1}{n_t} \cdot\big\langle \bfd^{t, \mathcal{C}} , \bfu^{f_s} \big\rangle = \frac{1}{n_t} \cdot\big\langle \bfd^{t, \mathcal{D}} , P_h^\top \bfu^{f_s} \big\rangle . 
\]
As $\bfd^{t, \mathcal{D}}$ is $O(\log^2 n)$-subgaussian by \Cref{lem:subg_dyadic}, it follows from Property \ref{property_2_lem:aug_trans_mat} in \Cref{lem:aug_trans_mat} that for any fixed outcome of the random shift $s$, $\big\langle \bfd^{t, \mathcal{D}} , P_h^\top \bfu^{f_s} \big\rangle$ is also subgaussian with parameter 
\[
O\big(\log^2 n \cdot \|P_h^\top \bfu^{f_s}\|_2^2\big) \leq O\big(\log^2 n \cdot \|\overline{P}_h^\top \bfu^{f_s}\|_2^2\big) .
\]
Thus our goal reduces to upper bounding  $\|\overline{P}_h^\top \bfu^{f_s}\|_2^2$. We start with the following claim. 
\begin{claim}[Coordinates of $\overline{P}_h^\top \bfu^{f_s}$]
\label{claim:coord_Pu_1D}
The vector $\overline{P}_h^\top \bfu^{f_s} \in \R^{|\widetilde{\D}_{\leq h}|-1}$ and its coordinates are given by $\big(\overline{P}_h^\top \bfu^{f_s}\big)_{\widetilde{I}} = f_s(r_{\widetilde{I}}) - f_s(\ell_{\widetilde{I}})$ for all dyadic intervals $\widetilde{I} \in \widetilde{\D}_{\leq h} \setminus [0,1)$.  Thus we have
\begin{align} \label{eq:l2_norm_1D}
\|\overline{P}_h^\top \bfu^{f_s}\|_2^2 = \sum_{\widetilde{I} \in \widetilde{\D}_{\leq h} \setminus [0,1)} \big(  f_s(r_{\widetilde{I}}) - f_s(\ell_{\widetilde{I}}) \big)^2 . 
\end{align}
\end{claim}

\begin{proof}
By Property \ref{property_1_lem:aug_trans_mat} in \Cref{lem:aug_trans_mat}, the ones in the columns of $\overline{P}_h$ correspond to the right-open dyadic intervals in $\widetilde{\D}_{\leq h}\setminus [0,1)$. So
for any $\widetilde{I} \in \widetilde{\D}_{\leq h} \setminus [0,1)$, we have 
\[
\big(\overline{P}_h^\top \bfu^{f_s}\big)_{\widetilde{I}} = \sum_{j: j/2^h \in \widetilde{I}} \bfu^{f_s}_j = \int_{\widetilde{I}} f_s'(z) d z = f_s(r_{\widetilde{I}}) - f_s(\ell_{\widetilde{I}}) .
\]    
This directly implies \eqref{eq:l2_norm_1D} and  proves the claim. 
\end{proof}
Now, we  express the $\ell_2$ norm $\E_s[\|\overline{P}_h^\top \bfu^{f_s}\|_2^2]$ in terms of the Fourier coefficients of $f$. 

\smallskip
\noindent \textbf{Bounding the $\ell_2$ Norm via Fourier.} 
Recall the Fourier expansion from \Cref{subsec:Fourier_analysis}, 
\[
f_s(z) = \sum_{k \in \Z}\widehat{f_s}(k) \cdot \exp(2 \pi i k z) =: \sum_{k \in \Z}  \widehat{f_s}(k) \cdot e_k(z),
\]
where $e_k(z)$ denotes the function $\exp(2 \pi i k z)$.
Also, by replacing $f_s$ with $f_s - \widehat{f_s}(0) = f_s - \widehat{f}(0)$ (this does not affect the integration error), we may assume wlog that $\widehat{f_s}(0) = 0$, so that $f_s(z) = \sum_{k \in \Z \setminus \{0\}}\widehat{f_s}(k) \cdot e_k(z)$. 
We need the following orthogonality property for different Fourier coefficients $\widehat{f_s}(k)$'s under the random shift $s$. 

\begin{lemma}[Fourier Orthogonality] \label{lem:orthogonality_fourier}
For any $k \neq k' \in \mathbb{Z}$, we have $\E_s\big[ \widehat{f_s}(k)^* \widehat{f_s}(k')\big] = 0$. 
\end{lemma}
\begin{proof}
Since $f_s(x) = f((x + s) \textsf{ mod } [0,1))$,  
we have
\[
\widehat{f_s}(k) = \int_0^1 f((x+s) \textsf{ mod } [0,1)) \cdot e_k(-x) d x = \int_0^1 f(x) \cdot e_k(-x + s) d x = e_k(s) \cdot \widehat{f}(k) .
\]
The lemma then follows as $\int_0^1 e_k(-s) \cdot e_{k'}(s) d s = 0$ by the orthogonality of the Fourier basis. 
\end{proof}

Given \Cref{lem:orthogonality_fourier}, it suffices to bound $|\overline{P}_h^\top \bfu^{e_k}|^2$ for each $k \in \mathbb{Z}$.  
\begin{lemma}[Bounding $\ell_2$ Norm for Each $k$] \label{lem:ell_2_bound_1D}
For any $k \in \mathbb{Z} \setminus \{0\}$, we have $|\overline{P}_h^\top \bfu^{e_k}|_2^2 \lesssim |k|$. 
\end{lemma}
\begin{proof}
We may assume $k > 0$ as the negative terms can be bounded similarly.  
By \Cref{claim:coord_Pu_1D} and the identity $|\exp(2 x i) - 1| = 2 |\sin(x)|$ for $x \in \R$,  
\begin{align*}
|\overline{P}_h^\top \bfu^{e_k}|_2^2 
& = \sum_{\widetilde{I} \in \widetilde{\D}_{\leq h} \setminus [0,1)} |e_k(r_{\widetilde{I}}) - e_k(\ell_{\widetilde{I}})|^2 = \sum_{j = 1}^h \sum_{I \in \widetilde{\D}_j} |e_k( r_{\widetilde{I}}) - e_k(\ell_{\widetilde{I}})|^2 \\
& = \sum_{j = 1}^h \sum_{\widetilde{I} \in \widetilde{\D}_j} |e_k (r_{\widetilde{I}} - \ell_{\widetilde{I}}) - 1|^2 
= 4 \sum_{j = 1}^h \sum_{\widetilde{I} \in \widetilde{\D}_j} \sin^2(\pi k (r_{\widetilde{I}} - \ell_{\widetilde{I}})) \\
& = 4 \sum_{j = 0}^h 2^j \cdot  \sin^2(\pi k / 2^j) .
\end{align*}
We break the sum over $j$ into ``long'' intervals with $j \leq \lfloor \log k \rfloor$, for which we bound $\sin^2(\pi k / 2^j) \leq 1$ trivially, and ``short'' intervals with $j > \lfloor \log k \rfloor$, for which we use the bound $\sin(x) \leq x$ whenever $x \geq 0$. Specifically, we have
\begin{align*}
|\overline{P}_h^\top \bfu^{e_k}|_2^2
& = 4 \sum_{j=0}^{\lfloor \log k \rfloor} 2^j + 4 \sum_{j =\lfloor \log k \rfloor + 1}^h 2^j \cdot (\pi k / 2^j)^2  \leq 8 (1 + \pi^2) k .
\end{align*}
This completes the proof of the lemma.  
\end{proof}

Combining \Cref{lem:orthogonality_fourier} and \Cref{lem:ell_2_bound_1D}, we obtain the bound (also recall \eqref{eq:sigma_SO_formula_intro})
\begin{align} \label{eq:ell_2_bound_1D}
\E_s\big[\|\overline{P}_h^\top \bfu^{f_s} \|_2^2\big] 
& = \E_s\Big[\Big\|\overline{P}_h^\top \Big(\sum_{k \in \Z \setminus \{0\}} \widehat{f_s}(k) \bfu^{e_k} \Big)\Big\|_2^2\Big] \nonumber\\
& =\sum_{k \in \Z \setminus \{0\}} \E\big[|\widehat{f_s}(k)|^2\big] \cdot \|\overline{P}_h^\top \bfu^{e_k} \|_2^2 
\lesssim \sum_{k \in \Z \setminus \{0\}} |\widehat{f}(k)|^2 \cdot |k| = \sigma_{\SO}(f)^2 .
\end{align}
where note that shifting $f$ only changes the phase of its Fourier coefficients but not their modulus.

Using the $\ell_2$-norm bound in \eqref{eq:ell_2_bound_1D} and \eqref{eq:disc_err_1D}, we can bound the variance of $\err^{\disc}(f_s)$ as
\begin{align} \label{eq:disc_var_1D}
\Var\big(\err^{\disc}(f_s)\big) 
& = \Var \Big( \sum_{t=0}^{T-1} \frac{1}{n_t} \cdot \big\langle \bfd^{t,\mathcal{C}} , \bfu^{f_s} \big\rangle \Big) \leq O\Big(\frac{\log^2 n}{n^2}\Big) \cdot \sigma_{\SO}(f)^2 ,
\end{align}
where the variance is over all the randomness (including the random shift $s$), and the last inequality follows from a standard martingale argument and that $n_t =n^22^{-t}$ is geometrically decreasing.

Now we are ready to put everything together and prove the 1-d case of \Cref{thm:main_high_dim}.

\begin{proof}[Proof of \Cref{thm:main_high_dim} in 1-d]
Plugging \eqref{eq:disc_var_1D} into \eqref{eq:error_decomp_1D}, we obtain that
\begin{align*}
\Var(\err(A_T,f_s)) & \leq \Var(\err(A_0,f_s)) +  \Var\big(\err^{\disc}(f_s)\big) \\
& \leq \frac{\sigma(f)^2}{n^2} + O\Big(\frac{\log^2 n}{n^2}\Big) \cdot \sigma_{\SO}(f)^2 \leq  \widetilde{O}(1) \cdot \Big(\frac{\sigma_{\mathsf{SO}}(f)^2}{n^2} \Big).
\end{align*}
This completes the proof of \Cref{thm:main_high_dim} for $d = 1$. 
\end{proof}

\subsection{Analysis for Higher Dimensions}
\label{sec:analysis_high_dim}
Now we give the full proof of \Cref{thm:main_high_dim}. Even though some parts are similar to the  1-d analysis in \Cref{sec:analysis_1D}, we repeat them again for clarity.

Fix an arbitrary 1-periodic function $f \in L^2([0,1]^d)$ with a continuous mixed derivative.
Again, we can equivalently view $\SubgSparsification$ as fixing $\D_{\leq h}^{\otimes d}$ and applying a random shift $\bfs \sim [0,1)^d$ to $f$ which results in $f_{\bfs}(\bfz) = f(\bfz + \bfs)$. Denote $\E_{\bfs}$ the expectation w.r.t. random shift $\bfs$.

By \Cref{lem:Hlawka_Zaremba_formula}, the integration error  
\begin{align*}
\err(A_T,f_{\bfs}) = \sum_{\emptyset \neq S \subseteq [d]} (-1)^{|S|-1} \int_{Q_S} h_T(\bfz) \cdot \partial_S f_{\bfs}(\bfz)  d \bfz ,
\end{align*}
where recall that $h_T(\bfz) = D(A_t, C_{\bfz})/n_t$ and $\partial_S f_{\bfs}(\bfz) = \frac{\partial^{|S|} f_{\bfs}(\bfz)}{\partial \bfz_S}$. Again, by \Cref{lem:decomp_cont_disc} we have,
\begin{equation} \label{eq:error_decomp_high_dim}
\begin{aligned} 
\err(A_T,f_{\bfs}) & = \err(A_0,f_{\bfs}) + \sum_{t=0}^{T-1} \err^{\disc}_t(f_{\bfs}) =:\err(A_0,f_{\bfs}) +  \err^{\disc}(f_{\bfs}) ,
\end{aligned}
\end{equation}
but this time the discrepancy error is given by
\begin{align}
\label{eq:disct-d}
    \err^{\disc}_t(f_{\bfs}) & =\sum_{\emptyset \neq S \subseteq [d]} (-1)^{|S|-1} \int_{Q_S} \frac{\disc_t(\bfz)}{n_t} \cdot \partial_S f_{\bfs}(\bfz) d \bfz .
\end{align}
\smallskip
\noindent \textbf{Bounding the Discrepancy Error.} 
Now we analyze the discrepancy error $\err^{\disc}_t$ in \eqref{eq:disct-d}.

At a high level, we analyze the contribution from each $\emptyset \neq S \subseteq [d]$ using tensorization, and then exploit the joint subgaussianity of discrepancy on each $Q_S$ to combine them. 

Let us fix $\emptyset \neq S \subseteq [d]$ and consider the integral in \eqref{eq:disct-d} for $S$.
Again, let us assume wlog that the combinatorial discrepancy $\disc_t(\bfz)$ on $Q_S$ is a step function with values changing only when any coordinate reaches $2^{-h} \Z \cap [0,1)$ (recall \Cref{footnote:perturb_disc}). 
This way, we may view $\disc_t(\bfz)$ restricted to $Q_S$ as a vector\footnote{Once again, we have ignored the $0$-measure set of corners $C \subseteq Q_S$ with some dimension being $(0,1]$.} $\bfd^{t, \mathcal{C}, S} \in \Z^{|\mathcal{C}_h^{\otimes S}|}$ with coordinates\footnote{Strictly speaking, the index $\bfj$ here should really range over $\{0, \cdots, 2^h-1\}^S \times \{2^h\}^{\overline{S}}$ so that the coordinates in $\overline{S}$ of $\bfj/2^h$ are $1$. But to keep notation clear, even though we are working with $d$-dimensional vectors we only display the coordinates of $\bfj$ in $S$ 
which are the only ones that can vary. The same convention applies when we talk about corners $\mathcal{C}_h^{\otimes S}$, dyadic boxes $\mathcal{D}_{\leq h}^{\otimes S}$, and decomposition $P_h^{\otimes S}$.} 
$\bfd^{t, \mathcal{C}, S}_{\bfj} = \disc_t(\bfj/2^h)$ for all $\bfj \in \{0, \cdots, 2^h-1\}^S$ corresponding to corners $\mathcal{C}_h^{\otimes S}$. 
We also define the vector $\bfu^{f_{\bfs},S} \in \R^{|\mathcal{C}_h^{\otimes S}|}$ as
\[
\bfu^{f_{\bfs},S}_{\bfj} := \int_{[\bfj/2^h, (\bfj+ \ind_S)/2^h] } \partial_S f_{\bfs}(\bfz_S, \ind_{\overline{S}}) \cdot d \bfz_S . 
\]
for all $\bfj \in \{0, \cdots, 2^h-1\}^S$. It is possible to further express $\bfu^{f_{\bfs},S}_{\bfj}$ as a sum of function values at the vertices of $[\bfj/2^h, (\bfj+ \ind_S)/2^h] \times \ind_{\overline{S}}$ with mixed signs, but we do not need such a formula. 
Next, observe that the vector $\bfu^{f_{\bfs},S}$ takes a tensor form if the function $f_{\bfs}$ is the product of functions of each coordinate. 

\begin{observation}[Tensor Form for Product Functions]
\label{obs:tensor_form_u}
If $f_{\bfs}: [0,1]^d \rightarrow \C$ has the form $f_{\bfs}(\bfz) = \prod_{i \in S} g_i(\bfz_i)$ for all $\bfz \in Q_S$, then we have $\bfu^{f_{\bfs},S}_{\bfj} = \prod_{i \in S} \bfu^{g_i}_{\bfj_i}$ for any $\bfj \in \{0, \cdots, 2^h-1\}^S$, where $\bfu^{g_i}$ is defined in \eqref{eq:u_f_vec}. Consequently, we have $\bfu^{f_{\bfs},S} = \otimes_{i \in S} \bfu^{g_i}$.
\end{observation}

This observation is immediate from the definition of $\bfu^{f_{\bfs},S}$ above.
As in the 1-d case, using dyadic decomposition, we may now express $\err_t^{\disc}$ as inner products
\begin{align} \label{eq:disc_err_high_dim}
\err^{\disc}_t(f_{\bfs})
& = \sum_{\emptyset \neq S \subseteq [d]} (-1)^{|S|-1} \int_{Q_S} \frac{\disc_t(\bfz)}{n_t} \cdot \partial_S f_{\bfs}(\bfz) d \bfz = \frac{1}{n_t} \sum_{\emptyset \neq S \subseteq [d]}  (-1)^{|S|-1} \cdot \big\langle \bfd^{t, \mathcal{C}, S} , \bfu^{f_{\bfs},S} \big\rangle \nonumber \\
& = \frac{1}{n_t} \sum_{\emptyset \neq S \subseteq [d]}  (-1)^{|S|-1} \cdot \big\langle \bfd^{t, \D, S} , (P_h^\top)^{\otimes S} \bfu^{f_{\bfs},S} \big\rangle ,
\end{align}
where $\bfd^{t, \D, S} \in \Z^{|(\D_{\leq h} \setminus C_1)^{\otimes S}|}$ is the restriction of $\bfd^{t,\D}$ (defined in \Cref{subsec:property_subg_sparse}) to dyadic boxes in $Q_S$ excluding those with any dimension in $S$ being $(0,1]$.  It follows that the vectors $\bfd^{t, \D, S}$ for different $\emptyset \neq S \subseteq [d]$ correspond to disjoint coordinates of $\bfd^{t,\D}$. Consequently, by \Cref{lem:subg_dyadic},  for any outcome of the random shift $\bfs$, $\err^{\disc}_t(f_{\bfs})$ is also subgaussian with parameter 
\begin{align}
\label{eq:disc_err_subg_him_dim}
O\Big(\frac{\log^{d+1} n}{n_t^2} \cdot \sum_{\emptyset \neq S \subseteq [d]} \big\|(P_h^\top)^{\otimes S} \bfu^{f_{\bfs},S} \big\|_2^2\Big) \leq O\Big(\frac{\log^{d+1} n}{n_t^2} \cdot \sum_{\emptyset \neq S \subseteq [d]} \big\|(\overline{P}_h^\top)^{\otimes S} \bfu^{f_{\bfs},S} \big\|_2^2\Big).
\end{align}

Therefore, we are left to derive an upper bound on  $\sum_{\emptyset \neq S \subseteq [d]} \big\|(\overline{P}_h^\top)^{\otimes S} \bfu^{f_{\bfs},S} \big\|_2^2$, for which we again look at the Fourier transformation of $f_{\bfs}$.

\smallskip
\noindent \textbf{Bounding the $\ell_2$ Norm via Fourier.} 
As in \Cref{subsec:Fourier_analysis}, we may write 
\[
f_{\bfs}(\bfz) = \sum_{\bfk \in \Z^d}\widehat{f_{\bfs}}(\bfk) \cdot e^{2 \pi i \langle \bfk , \bfz \rangle} =: \sum_{\bfk \in \Z^d}\widehat{f_{\bfs}}(\bfk) \cdot e_{\bfk}(\bfz) ,
\]
where $e_{\bfk}(\bfz) := \exp(2 \pi i \langle \bfk, \bfz\rangle)$. 
Again, by replacing $f_{\bfs}$ with $f_{\bfs} - \widehat{f}(\mathbf{0})$, we may assume wlog that $\widehat{f_{\bfs}}(\mathbf{0}) = 0$. Then we may write $\bfu^{f_{\bfs},S} = \sum_{\bfk \in \Z^d \setminus \mathbf{0}} \widehat{f_{\bfs}}(\bfk) \cdot \bfu^{e_{\bfk}, S}$. 
Note that the vector $\bfu^{e_{\bfk}, S}$ are the same for the set of $e_{\bfk}$ for which $\bfk_S$ coincide.

The following orthogonality property of the high dimensional Fourier coefficients is a simple corollary of its one dimensional version in \Cref{lem:orthogonality_fourier}.

\begin{corollary}[Fourier Orthogonality]
\label{cor:orthogonality_fourier_high_dim}
For any $\bfk \neq \bfk' \in \mathbb{Z}^d$, we have $\E_{\bfs}\big[ \widehat{f_{\bfs}}(\bfk)^* \widehat{f_{\bfs}}(\bfk')\big] = 0$. 
\end{corollary}

Now we can use \Cref{cor:orthogonality_fourier_high_dim,lem:ell_2_bound_1D} to compute the $\ell_2$ norm as follows.
\begin{align*}
\E_{\bfs} \Big[\Big\| (\overline{P}^\top_h)^{\otimes S} \cdot \bfu^{f_{\bfs},S} \Big\|_2^2\Big] 
& = \E_{\bfs} \Big[\Big| (\overline{P}^\top_h)^{\otimes S} \cdot \Big( \sum_{\bfk\neq \mathbf{0}} \widehat{f_{\bfs}}(\bfk) \bfu^{e_{\bfk},S} \Big) \Big|^2 \Big] \\
& = \sum_{\mathbf{k} \neq \mathbf{0}} \big | \widehat{f}(\bfk) \big|^2 \cdot \prod_{j \in S} |\overline{P}^\top_h \bfu^{e_{\bfk_j}} |^2 
\leq  \sum_{\bfk\neq \mathbf{0}} \big|\widehat{f}(\bfk) \big|^2  \cdot \prod_{j \in S} |\bfk_j|^2 ,
\end{align*}
where the expectation over $\bfs$ is removed in the second line since the modulus of the Fourier coefficients do not depend on the random shift. 

Now summing over all $\emptyset \neq S \subseteq [d]$ gives the following $\ell_2$ bound
(recall \eqref{eq:sigma_SO_formula_intro}).
\begin{align} \label{eq:ell_2_bound_high_dim}
\E_{\bfs} \Big[\sum_{\emptyset \neq S \subseteq [d]} \Big\| (\overline{P}^\top_h)^{\otimes S} \cdot \bfu^{f_{\bfs},S} \Big\|_2^2 \Big]
& \leq \sum_{\bfk\neq \mathbf{0}} \big|\widehat{f}(\bfk) \big|^2 \Big(\sum_{\emptyset \neq S \subseteq [d]}   \prod_{j \in S} |\bfk_j|^2 \Big) = O_d\big(\sigma_{\SO}(f)^2\big) .
\end{align}
Using the $\ell_2$-norm bound \eqref{eq:ell_2_bound_high_dim} in \eqref{eq:disc_err_high_dim} and \eqref{eq:disc_err_subg_him_dim}, we can bound $\E[(\err^{\disc}(f_{\bfs}))^2]$ as 
\begin{align} \label{eq:disc_var_high_dim}
\E\big[(\err^{\disc}(f_{\bfs}))^2\big] 
& = \E \Big[ \big(\sum_{t=0}^{T-1} \err^{\disc}_t(f_{\bfs})\big)^2 \Big] 
\leq O_d\Big(\frac{\log^{d+1} n}{n^2} \cdot \sigma_{\SO}(f)^2 \Big).
\end{align}
Now we have everything we need to prove \Cref{thm:main_high_dim}. 
\begin{proof}[Proof of \Cref{thm:main_high_dim}]
Plugging \eqref{eq:disc_var_high_dim} into \eqref{eq:error_decomp_high_dim}, we obtain that
\begin{align*}
\E[(\err(A_T,f_{\bfs}))^2] & \leq \E[(\err(A_0,f_{\bfs}))^2] + \E[\big(\err^{\disc}(f_{\bfs})\big)^2]  \\
& \leq  \frac{\sigma(f)}{n^2}+ O_d\Big(\frac{\log^{d+1} n}{n^2} \cdot \sigma_{\SO}(f)^2 \Big)  \leq \widetilde{O}_d\Big(\frac{\sigma_{\SO}(f)^2}{n^2}\Big) .
\end{align*}
This completes the proof of the theorem. 
\end{proof}

\subsection{General Definition for $\sigma_{\SO}$} 
\label{subsec:general_def_sigma_SO}
We have thus far assumed that $f$ has a Fourier series expansion and have defined $\sigma_{\SO}$ in terms of the Fourier coefficients of $f$, but this is not really necessary. 
Inspecting the proof of \Cref{thm:main_high_dim} more closely reveals an intrinsic definition for $\sigma_{\SO}$ which does not rely on the Fourier coefficients of $f$. This allows our results to be applicable to more general integrable functions. 

In particular, the general form for $\sigma_{\SO}$ is given by
\begin{align} \label{eq:sigma_SO_general}
\sigma_{\SO}(f)^2 := \lim_{h \rightarrow \infty} \E_{\bfs} \Big[\sum_{\emptyset \neq S \subseteq [d]} \Big\| (\overline{P}^\top_h)^{\otimes S} \cdot \bfu^{f_{\bfs},S} \Big\|_2^2 \Big] ,
\end{align}
where the expectation is taken over the random shift $\bfs$ of the function $f$. 
To illustrate this definition, let us consider 1-d, where \eqref{eq:sigma_SO_general} simplifies to
\[
\sigma_{\SO}(f)^2 := \lim_{h \rightarrow \infty} \E_s\Big[\big\|\overline{P}_h^\top \bfu^{f_s} \big\|_2^2\Big] .
\]
Note that the vector of derivatives $\bfu^{f_s}$ depends both on $f_s$ and the granularity $2^{-h}$. 
Equivalently, $\overline{P}_h^\top \bfu^{f_s}$ is the vector whose coordinates correspond to the integration of $f_s$ over all dyadic boxes with edge length at least $2^{-h}$. 
Consequently, the quantity $\big\|\overline{P}_h^\top \bfu^f \big\|_2$ increases monotonically with $h$ and thus $\sigma_{\SO}(f)$ in \eqref{eq:sigma_SO_general} is well-defined (with value $\infty$ if the sequence diverges). 
When $f$ has a Fourier series, following  \eqref{eq:ell_2_bound_high_dim}, the definition of $\sigma_{\SO}(f)$ in \eqref{eq:sigma_SO_general} is at most the Fourier definition in  \eqref{eq:sigma_SO_formula_intro} (and can actually be much smaller). 
So in a sense, the definition of $\sigma_{\SO}$ in \eqref{eq:sigma_SO_general} is the fundamental one, and we chose to work with the Fourier definition \eqref{eq:sigma_SO_formula_intro} only to have an explicit formula. 

\medskip 
\noindent {\bf Comparison with $V_{\HK}(f)$.} The general definition of $\sigma_{\SO}(f)$ in \eqref{eq:sigma_SO_general}  compares favorably against the $\ell_1$-version of Hardy-Krause variation $V_{\HK}(f)$ and can be significantly better. 
Let us demonstrate this comparison in $1$-d.
Here, recall that $V_{\HK}(f)$ is given by  $\int_0^1 |f'(x)| d x$, which equivalently is
\[
\sup_{0 = p_0 \leq \cdots \leq p_m = 1} \sum_{j=1}^m |f(p_j) - f(p_{j-1})|,
\]
where the supremum is over all sequences of points $0 = p_0 \leq \cdots \leq p_m = 1$.

Now consider every level $j \geq 0$ of dyadic intervals $I_1, \cdots, I_{2^j}$. The coordinates of $\overline{P}_h^\top \bfu^f$ for level-$j$ dyadic intervals are $f(r_{I_a}) - f(\ell_{I_a})$, where $a = 1, \cdots, 2^j$, and these correspond to the partition of the interval $[0,1]$ by level-$j$ dyadic intervals. 
Note that the $\RHS$ of \eqref{eq:sigma_SO_general} is taking the $\ell_2$-norm of these coordinates, while $V_{\HK}(f)$ is taking the $\ell_1$-norm of the values $|f(p_j) - f(p_{j-1})|$ in the partition. Consequently, there is a geometric decay with $j$, and we immediately obtain $\sigma_{\SO}(f) \leq O(V_{\HK}(f))$. 

Not only that, but as before, since the high-frequency parts of $f$ admit cancellations along long dyadic intervals, the variation $\sigma_{\SO}(f)$ is substantially smaller than $V_{\HK}(f)$ when $f$ has a large high-frequency component. So the resulting bound then in \Cref{thm:main_high_dim} is significantly superior to the $\ell_1/\ell_\infty$ version of the Koksma-Hlawka inequality. 
This is precisely the same phenomenon that has been exploited in the proof of \Cref{thm:main_high_dim} -- just that we no longer have clean and explicit formulas for the quantities involved!

\section{Concluding Remarks}
\label{sec:conclude}
In this work, we presented a simple yet powerful randomized algorithm, $\SubgSparsification$, and showed that it produces QMC sets that go beyond the Hardy-Krause variation and substantially improve over the classical Koksma-Hlawka inequalities. As we discussed, the algorithm also possesses many other desirable features, e.g., it achieves amortized $\widetilde{O}_d(1)$ time per point; it uses random samples and is thus as flexibility as MC; it automatically achieves the best of MC and QMC (and the improvement over Hardy-Krause and Koksma-Hlawka) in the optimal way. 

\smallskip
\noindent \textbf{An Open Problem.} Perhaps one slight downside of our algorithm is that it starts with $n^2$ samples, and then it splits them into $n$ QMC sets of size $n$ each in amortized $\widetilde{O}_d(1)$ time per sample. It would be nice to obtain a more direct algorithm that achieves our improvement over Hardy-Krause variation and Koksma-Hlawka inequalities, but 
only needs $\widetilde{O}_d(n)$ random samples to generate a $n$-point QMC set, and runs in $\widetilde{O}_d(n)$ time.

For constructing low-discrepancy point sets, this was done in \cite{DFG+19}, who gave an elegant algorithm that uses $2n$ random samples to produce an $n$-point set with $\widetilde{O}_d(1)$ discrepancy in $\widetilde{O}_d(n)$ time.\footnote{In contrast, using the transference principle to achieve this would also require starting with $n^2$ samples.} 
However, obtaining such an algorithm in our context seems significantly more challenging. The main bottleneck for approaches based on the transference principle is that the error term $\err(A_0,f)$ due to the initial samples in \eqref{eq:error_decomp_high_dim} becomes $\Omega(\sigma(f)/\sqrt{n})$ when one uses only $O(n)$ samples, which is as large as the MC error. We leave the question of obtaining such an algorithm that uses fewer samples as an interesting open problem.

\newpage

\appendix

\section{Further Related Work}
\label{subsec:related_work}

\smallskip \noindent \textbf{Low-Discrepancy Constructions.} Low-discrepancy constructions are central objects of study in QMC and other related fields.  
For intuition, consider $d=1$.
For $n$ random points, the   $D^*(X) = \Omega(\sqrt{n})$ typically (already the sub-interval $[0,1/2]$ will typically have $n/2 \pm \Omega(\sqrt{n})$ points). 
However here, simply choosing equally spaced points at distance $1/n$,  trivially satisfies that $D^*(X)=1$.

Perhaps surprisingly, analogous constructions fail badly for $d=2$, e.g., 
uniformly spaced gridpoints in the square have star discrepancy $\Omega(\sqrt{n})$ (see \Cref{fig:square}); essentially no better than random points.

\begin{figure}[hbtp!]
   \centering
\includegraphics[scale=0.5, angle=270]
{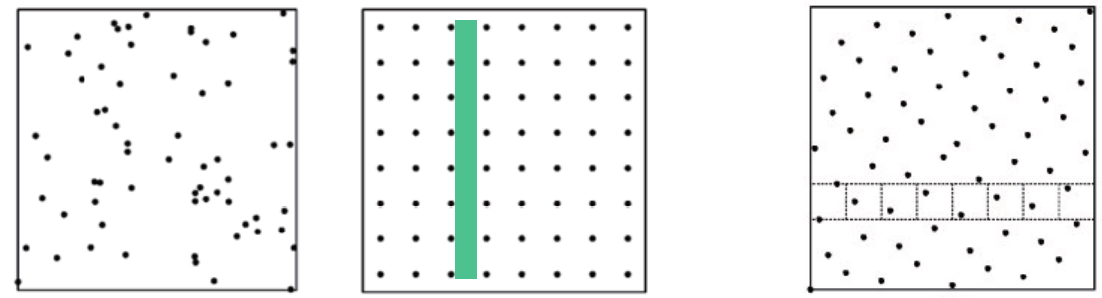}
     \caption{\small{Placements of 64 points in $[0,1]^2$ (from \cite{Mat09}). On the left are random points. In the middle, the grid points have large discrepancy as the green rectangle has area $1/\sqrt{n}$ but no points. On the right is the Van der Corput set, which already visually looks more ``uniform".}}
     \label{fig:square}
 \end{figure}

Remarkably, point sets with $O(\log n)$ discrepancy exist based on careful and beautiful explicit constructions! One such example is the van der Corput set (see \Cref{fig:square}), first discovered in 1935 \cite{vdC35a,vdC35b}.

Since then, there has been extensive work on designing such low-discrepancy sets in higher dimensions. Several ingenious and elegant constructions have been discovered, often involving deep  mathematical techniques.
Some well-known examples include Halton-Hammersley sets \cite{Ham60, Hal60}, Sobol' sequences \cite{Sob67}, Faure sequences \cite{Fau82}, digital nets \cite{Nie92}, lattice rules \cite{Kro59}, etc. 

The current best known bounds are a tight $\Theta(\log^{(d-1)/2} n)$ for the $\ell_2$-discrepancy for corners $D_2^*(X)$, and a worse upper bound of $O(\log^{d-1} n)$ for the $\ell_\infty$-discrepancy $D^*(X)$. It remains a ``great open problem'' \cite{BC87} to improve the bound on $D^*(X)$ for $d > 2$ (for $d=2$, a tight lower bound of $\Omega(\log n)$ is known). 
Some excellent references are \cite{Nie92, Gla04, Mat09,Owe13}.

\smallskip \noindent \textbf{Randomized QMC methods.} 
As discussed in \Cref{sec:intro}, approaches to use randomization with QMC to bypass the limitations of QMC methods have been explored. 
Perhaps the simplest way is to apply a (single) random shift to $X$, where each point $\bfx \in X$ is shifted by the {\em same} amount $\bfs$ to the point $\bfx+ \bfs$ (modulo $[0,1]^d$). This ensures that any single point $\bfx+ \bfs$ is uniformly random, and it does not change the discrepancy of $X$. However, the randomness here is very limited to achieve the guarantees of MC, e.g., there is no pairwise randomness.

Another simple randomization method is to perturb each points of a digital net inside its corresponding square lattice cell uniformly at random (see \Cref{fig:square} (right)). This method is known as {\em jittered sampling}. While jittered sampling performs well in 1-d, such randomization destroys the benefits of QMC quickly for $d\geq 2$. Already, in $2$-d the discrepancy degrades to about $n^{1/4}$, as a typical corner/rectangle will intersect $n^{1/2}$ cells. For general $d$, the discrepancy of jittered sampling becomes $n^{1/2-1/2d}$.

More sophisticated randomized QMC methods have been proposed, where explicit QMC constructions, such as digital nets, are perturbed in much more delicate ways and one considers complicated explicit distributions over these perturbed sets. 
Some well-studied 
methods include the Cranley-Patterson rotation of lattice constructions \cite{CP76}, digital shifted nets \cite{LL02}, scrambled nets of Owen \cite{Owe95,Owe97a,Owe97b} and their variants \cite{Hic96,Mat98c}, etc.

Variants of scrambled nets \cite{Owe97a,Mat98c} can also achieve the best of both QMC and MC guarantees, similar to \Cref{thm:best-of-both-1}.
However, these methods suffer from the limitations of having to start with an explicit initial set of QMC points, and do not go beyond the best of both QMC and MC guarantees. 
One exception is the work by Owen \cite{Owe97b,Owe08}, though the focus there is very different from this paper. In particular, Owen showed a better convergence rate of $O(n^{-3/2})$ for scrambled nets, using a different variation measure (which can be arbitrarily larger than the Hardy-Krause variation) and additional smoothness assumptions.   
We refer interested readers to the excellent book \cite{Owe13} and the survey \cite{LL02} for more details on randomized QMC methods. 

\smallskip \noindent \textbf{Transference Principle and QMC constructions.}
The transference principle has been long known (and is often attributed to Vera Sos), see e.g.,~\cite{Mat09, AistleitnerBN17}. Consider the following combinatorial discrepancy problem called Tusnady's problem. Given a set of $n$ (arbitrary) points in $[0,1]^d$, color them $\pm 1$ to minimize the discrepancy of axis-parallel boxes.
The transference principle says that an $\alpha(n)$ bound for Tusnady's problem,  implies a placement of $n$ points starting with $n^2$ random points, with continuous discrepancy $O(\alpha(n))$. The reduction is algorithmic (as used in this paper). 

Until recently though, this connection was mostly of theoretical interest as most of the better bounds for combinatorial discrepancy problems were based on non-constructive proof techniques. In recent years there has been a lot of progress in this direction,  resulting in several algorithmic approaches such as  \cite{Ban10, lovettmeka, Ro14, EldanS18, bdgl18, ALS21, PesentiV23, JambulapatiRT24} that match the non-constructive bounds.  Several of these are also very efficient.
Remarkably, the best known non-constructive bound for Tusnady's problem (and hence for star discrepancy using these methods) is $O(\log^{d-1/2}n)$ \cite{Nikolov17}, which almost matches the best known bound of $O(\log^{d-1} n)$ based on explicit constructions. The best algorithmic bound for Tusnady's problem is $O(\log^d n)$ \cite{BansalG17, ALS21}.

Finally, we mention the beautiful algorithm due to Dwivedi et al.~\cite{DFG+19}, that finds a $n$-point set with discrepancy $\widetilde{O}_d(1)$ in $[0,1]^d$, using only $O(n)$ random samples. See also  \Cref{sec:conclude}.

\section{Derivation of the Hlawka-Zaremba Formula in 1-D}
\label{sec:appendix_QMC}

In the 1-d case, the Hlawka-Zaremba formula (\Cref{lem:Hlawka_Zaremba_formula}) can be derived by a simple application of integration by parts. Let us denote the points in $A$ as $0 \leq x_1 \leq \cdots \leq x_n \leq 1$ and set $x_{n+1}:=1$.  Integration by parts gives 
\[\int_0^1 f(x) dx = \big[xf(x) \big]_0^1- \int_0^1 xf'(x) dx  = f(1) - \int_0^1 xf'(x) dx.\]
Also, $\sum_{i=1}^n i (f(x_i) - f(x_{i+1})) = \sum_{i=1}^n f(x_i) - nf(x_{n+1})$.

Thus the  integration error of $A$ w.r.t. $f$ can be written as 
\begin{align} \label{eq:Koksma_Hlawka_techniques}
\err(A,f) 
& := \frac{1}{n}\sum_{i=1}^n f(x_i) - \int_0^1 f(x) d x \nonumber \\
& = \Big(f(x_{n+1}) - \frac{1}{n} \sum_{i=1}^n i (f(x_{i+1}) - f(x_i)) \Big) - \Big( f(x_{n+1}) - \int_0^1 x \cdot f'(x) d x \Big) \nonumber \\
& = \int_0^1 x \cdot f'(x) d x - \sum_{i=1}^n \int_{x_i}^{x_{i+1}} \frac{i}{n} \cdot  f'(x) d x 
= \int_0^1 h(x) \cdot f'(x) d x ,
\end{align}
where the last line splits the integral $\int_0^1$ as the sum of integrals  $\sum_{i=1}^n \int_{x_i}^{x_{i+1}}$, and finally that
$ h(x):= x - |A \cap [0,x]|/n$
is the (scaled) continuous discrepancy of the prefix interval $[0,x]$ w.r.t. the point set $A$. 

In fact, the higher dimensional case of \Cref{lem:Hlawka_Zaremba_formula} can be derived similarly using integration by parts, but the formula gets much more complicated (due to boundary conditions).

\section{Failure of Subgaussianity for Prefix Intervals and Corners}
\label{sec:example_subg_fail}

We start with a simple example to demonstrate that subgaussianity for the combinatorial discrepancy of prefix intervals $\bfd^{\mathcal{C}}$ may fail even when it holds for the dyadic intervals $\bfd^{\D}$. In particular, we describe a specific distribution where $\bfd^{\mathcal{C}}$ may be $\Omega(n)$-subgaussian, even though $\bfd^{\D}$ is $\widetilde{O}(1)$-subgaussian. Moreover, in this example, $\|\bfd^{\mathcal{C}}\|_\infty = \widetilde{O}(1)$. This illustrates that the subgaussian property is much more delicate and does not simply follow from having low discrepancy. 

Consider the dyadic intervals $I_1 := [0,1/2n]$, $I_2 := [1/2n,1/n]$, $\cdots$, $I_{2n} := [(2n-1)/2n, 1]$ of length $1/2n$. Suppose for convenience that they  contain exactly one point from $A_{T-1}$. Consider the distribution over colorings where
the colors for consecutive points (except those in $I_1$ and $I_{2n}$) are paired, i.e.,
$\disc(I_2) = - \disc(I_3)$, $\disc(I_4) = - \disc(I_5)$, $\cdots$, $\disc(I_{2n-2}) = \disc(I_{2n-1})$,  and that the random variables, $\disc(I_2), \disc(I_4), \cdots, \disc(I_{2n-2})$  and for the other two intervals $\disc(I_1), \disc(I_{2n})$ are i.i.d.~$\pm 1$. 

It is easily verified that the vector $\bfd^{\D}$ consisting of discrepancy of all dyadic intervals  is $\widetilde{O}(1)$-subgaussian (this is because $\bfd^{\D}$ restricted to intervals of length $1/2n$ is clearly $O(1)$-subgaussian, and for dyadic intervals of length more than $1/2n$, any interval that does not contain $I_1$ or $I_{2n}$ has discrepancy $0$, and only $O(\log n)$ intervals contain $I_1$ or $I_{2n}$).

Nonetheless, if we consider the prefix intervals, observe that we have $\disc(C_{(2j - 1)/2n}) = \disc(I_1)$ for all $j = 1, \cdots, n$. That is, these $n$ coordinates of $\bfd^{\mathcal{C}}$ are all either $1$ or $-1$ simultaneously and hence are completely {\em correlated}. Consequently, the subgaussian constant of  $\mathbf{d}^{\mathcal{C}}$ is  $\Omega(n)$.

As is pointed out in \Cref{foonote:subg_fail}, however, in the 1-d case, one can design a different distribution over colorings where $\bfd^{\mathcal{C}}$ is actually $O(1)$-subgaussian.
In particular, one can pair the points from left to right, assign i.i.d. uniformly random colors $\bfx_{2i+1} \in \{\pm 1\}$ to the $(2i+1)$th point for each $i = 0, \cdots, n/2-1$, and then set $\bfx_{2i+2} = - \bfx_{2i+1}$. It is immediate to verify that the resulting discrepancy vector of prefix intervals $\bfd^{\mathcal{C}}$ is indeed $O(1)$-subgaussian. 

\smallskip
\noindent \textbf{A stronger $\Omega(n)$ Subgaussianity Lower Bound for $\bfd^{\mathcal{C}}$ when $d \geq 2$.}
Given the above strategy for the $1$-d case, one may  naturally wonder if $\widetilde{O}_d(1)$-subgaussianity for $\bfd^{\mathcal{C}}$ might be achievable more generally for $d \geq 2$, by designing some specific distribution. 

Somewhat surprisingly however, we show  that this is impossible in a very strong sense. More specifically, we first construct an example in dimension $d=2$ where the discrepancy of corners $\bfd^{\mathcal{C}}$ must be $\Omega(n)$-subgaussian for {\em any} distribution over colorings. Using a similar argument, we can in fact show the very strong lower bound of $\Omega(n)$ on the subgaussian constant of  $\bfd^{\mathcal{C}}$ even for a uniformly random placement of $n$ points in $(0,1]^2$. Note that these lower bounds also hold for $d \geq 2$ by restricting to corners with edges $(0,1]$ in the remaining $d-2$ dimensions. 

Consider points $\bfz^1, \cdots, \bfz^n \in (0,1]^2$, where $\bfz^1= (1, 1/n)$ and $\bfz^i = (\bfz^i_1, i/n)$ with $\bfz^i_1 \in (0, 1 - 1/n]$ for all $i = 2, \cdots, n$.
Let $\bfx \in \{\pm 1\}^n$ be an arbitrary (random) coloring of these points. 
For each $i \in [n]$, consider the pair of corners $C_{i,0} := (0,1 - 1/n] \times (0, i/n]$ and $C_{i,1} := (0,1] \times (0, i/n]$. 
Note that $C_{i,0}$ contains points $\bfz^2, \cdots, \bfz^i$, while $C_{i,1}$ contains points $\bfz^1, \cdots, \bfz^i$. Consequently, $\disc(C_{i,1}) - \disc(C_{i,0}) = \bfx_1$ for all $i \in [n]$, which implies that
\[
\Big|\sum_{i \in [n]} \big(\disc(C_{i,1}) - \disc(C_{i,0}) \big)\Big| = n . 
\]
This shows that the $2n$ coordinates of $\bfd^{\mathcal{C}}$ corresponding to $C_{i,0}$ and $C_{i,1}$ for all $i \in [n]$ must be at least $\Omega(n)$-subgaussian. 

Note that this $\Omega(n)$ lower bound on the subgaussianity of $\bfd^{\mathcal{C}}$ continues to hold even when the points $\bfz^1, \cdots, \bfz^n$ are chosen i.i.d. uniformly at random from $(0,1]^2$. In particular, let $\bfz^i$ be the point with the largest $\bfz^i_1$, then with high probability, there are $\Omega(n)$ points $\bfz^j$ with $\bfz^j_2 > \bfz^i_2$. Applying the same argument as above to this set of points and $\bfz^i$ prove an $\Omega(n)$ lower bound on the subgaussian constant of $\bfd^{\mathcal{C}}$ for i.i.d. uniformly random points. 

\section{Missing Details in \Cref{subsec:dyadic_decomp}}
\label{sec:missing_proofs}

In this section, we give more details on the dyadic decomposition matrix $P_h$, the structured decomposition matrix $\overline{P}_h$, and prove \Cref{lem:aug_trans_mat} (restated below). 

\AugDecompMatrix*

To prove \Cref{lem:aug_trans_mat}, let us first fix some terminology. Among the $2^j$ dyadic intervals of non-zero level $j > 0$, we call the ones of the form $(2\ell/2^j, (2\ell+1)/2^j]$ {\em odd}, and the other ones {\em even}. 
We emphasize that our restriction of the notion of odd and even intervals to non-zero levels is for convenience. 
We also call the interval $(2\ell/2^j, (2\ell+1)/2^j]$ the {\em sibling} of $(2\ell + 1/2^j, (2\ell+2)/2^j]$, and vice versa. 
Odd and even intervals correspond to odd and even numbers if we count the intervals in $\D_j$ from left to right. 

We use the following observation about the structures of columns of $P_h$. 

\begin{observation}[Dyadic Decomposition Matrix] \label{lem:dyadic_decomp_1D}
For any $h \in \Z_{>0}$, the dyadic decomposition matrix $P_h \in \{0,1\}^{|\mathcal{C}_h| \times (|\D_{\leq h}|-1)}$ defined above satisfies the following properties.
\begin{enumerate}
    \item The number of ones in every non-zero column of $P_h$ is a power of $2$. The zero columns of $P_h$ correspond to all the even dyadic intervals in $\D_{\leq h}$.  
    
    \item For every integer $0 \leq \ell \leq h-1$, there are exactly $2^\ell$ columns of $P_h$ that contains exactly $2^{h-\ell-1}$ ones, and these are exactly the odd dyadic intervals in $\D_{\ell+1}$.  
    
    \item  For any column of $P_h$ with $2^{h-\ell-1}$ ones, the ones are at the consecutive prefix intervals $C_{  \frac{2j+1}{2^{\ell+1}}}, \cdots, C_{  \frac{2j+2}{2^{\ell+1}} - \frac{1}{2^h}}$, where each $j \in \{0, \cdots, 2^\ell-1\}$ corresponds to one such column. 
\end{enumerate}
\end{observation}

In Property 3 of \Cref{lem:dyadic_decomp_1D}, the set of endpoints $\{\frac{2j+1}{2^{\ell+1}}, \cdots, \frac{2j+2}{2^{\ell+1}} - \frac{1}{2^h}\}$ can be viewed as all grid points of $\Z/2^h$ that lie in the right-open dyadic interval $[\frac{2j+1}{2^{\ell+1}}, \frac{2j+2}{2^{\ell+1}}) \in \widetilde{\D}_{\ell+1}$. Over all $j \in \{0, \cdots, 2^j-1\}$, these are exactly the even dyadic intervals of $\widetilde{\D}_{\ell+1}$. Consequently, the ones in the columns of $P_h$ can be equivalently viewed as all the even right-open dyadic intervals of $\widetilde{\D}_{\leq h}$. 

To simplify our analysis, we define the structured decomposition matrix $\overline{P}_h \in \{0,1\}^{|\mathcal{C}_h| \times (|\D_{\leq h}|-1)}$ so that the ones in the columns of $\overline{P}_h$ correspond to {\em all} dyadic intervals of $\widetilde{\D}_{\leq h}$. 

\begin{definition}[Structured Decomposition Matrix] \label{defn:aug_trans_mat}
The matrix $\overline{P}_h \in \{0,1\}^{|\mathcal{C}_h| \times (|\D_{\leq h}|-1)}$ is defined as follows:
\begin{enumerate}
    \item For each odd dyadic interval $I \in \D_{\leq h}$, let the column $(\overline{P}_h)_{\cdot, I} = (P_h)_{\cdot, I}$. 
    \item For each even dyadic interval $I \in \D_{\leq h}$, let $\widetilde{I} \in \widetilde{\D}_{\leq h}$ be the right-open sibling of $I$ (so $\widetilde{I}$ is an odd dyadic interval). Then let $(\overline{P}_h)_{C_{z}, I} = 1$ if $z \in 2^{-h}\Z \cap \widetilde{I}$, and $0$ otherwise
\end{enumerate}
\end{definition}

From \Cref{defn:aug_trans_mat} and \Cref{lem:dyadic_decomp_1D}, one can see that the structured decomposition matrix $\overline{P}_h$ is obtained from $P_h$ by replacing its zero columns with non-zero columns with the positions of ones corresponding to even right-open dyadic intervals.
\Cref{lem:aug_trans_mat} now follows immediately.

\begin{proof}[Proof of \Cref{lem:aug_trans_mat}]
The first property follows from the definition of $\overline{P}_h$ and \Cref{lem:dyadic_decomp_1D}. 
By \Cref{lem:dyadic_decomp_1D} and \Cref{defn:aug_trans_mat}, every non-zero column of $P_h$ also appears in $\overline{P}_h$, from which the second property follows immediately. 
\end{proof}

In light of \Cref{lem:aug_trans_mat}, it is actually more appropriate to think of the columns of $\overline{P}_h$ as indexed by $\widetilde{\D}_{\leq h} \setminus \{C_1\}$ as opposed to $\D_{\leq h} \setminus \{C_1\}$, and this view is helpful for the analysis.

\section*{Acknowledgements}
Part of the work was done while the second-named author was a Postdoctoral Researcher in the Algorithms Group at Microsoft Research Redmond, and the first-named author was visiting the Simons Institute for the Theory of Computing at UC Berkeley.

\bibliographystyle{alpha}
\bibliography{bib.bib}

\end{document}